\documentclass[a4paper,english]{elsarticle}

\usepackage[T1]{fontenc}
\usepackage[utf8]{inputenc}
\usepackage{microtype}

\usepackage{mathtools}
\usepackage{amsfonts,amsmath,amssymb,amsthm}
\usepackage{xspace,xcolor}
\usepackage[basic,classfont=caps]{complexity}
\usepackage{hyperref}
\usepackage{todonotes}
\usepackage{tikz}
\usetikzlibrary{arrows,shapes,automata,
calc,chains,matrix,positioning,scopes,fit,backgrounds}
\usepackage{paralist}
\usepackage{multicol}
\usepackage{pgfplots}
\pgfplotsset{compat=1.11}

\theoremstyle{plain}
\newtheorem{theorem}{Theorem}
\newtheorem{lemma}[theorem]{Lemma}
\newtheorem{corollary}[theorem]{Corollary}
\newtheorem{proposition}[theorem]{Proposition}
\newtheorem{fact}[theorem]{Fact}
\newtheorem{problem}[theorem]{Problem}
\newtheorem{claim}[theorem]{Claim}

\theoremstyle{remark}
\newtheorem{remark}{Remark}
\newtheorem{example}{Example}

\renewcommand{\epsilon}{\varepsilon}
\newcommand{\defequals}{\stackrel{\mathrm{def}}{=}}
\providecommand{\st}{}
\renewcommand{\st}{\mathrel{\mid}}

\newcommand{\uu}{\boldsymbol{u}}
\newcommand{\vv}{\boldsymbol{v}}
\newcommand{\xx}{\boldsymbol{x}}
\newcommand{\yy}{\boldsymbol{y}}
\newcommand{\hh}{\boldsymbol{h}}

\renewcommand{\aa}{\boldsymbol{a}}
\newcommand{\bb}{\boldsymbol{b}}
\renewcommand{\cc}{\boldsymbol{c}}
\newcommand{\dd}{\boldsymbol{d}}

\newcommand{\rr}{\boldsymbol{r}}
\renewcommand{\ss}{\boldsymbol{s}}
\newcommand{\pp}{\boldsymbol{p}}
\newcommand{\qq}{\boldsymbol{q}}
\newcommand{\vspan}{\mathrm{span}}

\newcommand{\Aa}{\mathcal{A}}
\newcommand{\Bb}{\mathcal{B}}
\newcommand{\Cc}{\mathcal{C}}

\newcommand{\inter}[1]{[\![#1]\!]}
\newcommand{\funcA}{\inter{\Aa}}
\newcommand{\funcB}{\inter{\Bb}}

\newcommand{\accruns}[2]{\mathrm{Acc}_{#1}(#2)}

\newcommand{\periods}{\mbox{Periods}}

\newcommand{\distr}{\mathcal{D}}
\newcommand{\prob}[2]{\mathrm{Pr}_{#1}(#2)}

\newcommand{\nat}{\mathbb{N}}
\renewcommand{\int}{\mathbb{Z}}
\newcommand{\rat}{\mathbb{Q}}
\newcommand{\real}{\mathbb{R}}

\tikzset{
	every state/.style={inner sep=0.3mm, minimum size=0.5cm,
                        draw=blue!50,very thick,fill=blue!20,scale=0.8},
	accepting/.style={accepting by arrow},
	trans/.style={->, >=latex}
}

\journal{Journal of Computer and System Sciences}

\begin{document}

\begin{frontmatter}

\title{When are Emptiness and Containment Decidable \\ for Probabilistic Automata?}

\author[warwick]{Laure Daviaud}
\ead{L.Daviaud@warwick.ac.uk}
\author[warwick]{Marcin Jurdzi\'nski}
\ead{Marcin.Jurdzinski@warwick.ac.uk}
\author[warwick]{Ranko Lazi\'{c}}
\ead{R.S.Lazic@warwick.ac.uk}
\author[bordeaux]{Filip Mazowiecki}
\ead{filip.mazowiecki@u-bordeaux.fr}
\author[brussels]{Guillermo A. P\'erez}
\ead{gperezme@ulb.ac.be}
\author[oxford]{James Worrell}
\ead{James.Worrell@cs.ox.ac.uk}

\address[warwick]{University of Warwick, Coventry, UK}
\address[bordeaux]{Universit\'e de Bordeaux, Bordeaux, France}
\address[brussels]{Universit\'e libre de Bruxelles, Brussels, Belgium}
\address[oxford]{University of Oxford, Oxford, UK}

\begin{abstract}
    The emptiness and containment problems for probabilistic automata are
    natural quantitative generalisations of the classical language emptiness and
    inclusion problems for Boolean automata.  It is well known that both
    problems are  undecidable.  In this paper we provide a more refined view of
    these problems in terms of the degree of ambiguity of probabilistic
    automata.  We show that a gap version of the emptiness problem (that is
    known be undecidable in general) becomes decidable for automata of
    polynomial ambiguity.  We complement this positive result by showing that
    the emptiness problem remains undecidable even when restricted to automata
    of linear ambiguity.  We then turn to finitely ambiguous automata.  Here we
    show decidability of containment in case one of the automata is assumed to
    be unambiguous while the other one is allowed to be finitely ambiguous.  Our
    proof of this last result relies on the decidability of the theory of real
    exponentiation, which has been shown, subject to Schanuel's Conjecture, by
    Macintyre and Wilkie.  
\end{abstract}


\end{frontmatter}

\section{Introduction}
\label{sec:introduction}
Probabilistic automata (PA) are a quantitative extension of classical Boolean
automata that were first introduced by Rabin~\cite{rabin63}.  In this model, 
non-deterministic choices are replaced by probabilities: each transition carries 
a rational number which gives its probability to be chosen amongst all the other transitions going
out of the same state and labelled by the same letter.  Then, instead of simply
accepting or rejecting a word, such an automaton measures the probability of it
being accepted. 

PA can be seen as (blind) partially observable Markov
decision processes~\cite{puterman05}. PA are also closely related 
to hidden Markov models, which are finite-state models for generating probability
distributions over strings~\cite{baum1966}.  Such probabilistic finite-state machines have numerous applications in
the field of artificial intelligence~\cite{rn10,learning-survey}. Further
applications for PA include, amongst others, verification of
probabilistic systems~\cite{vardi85,knpq10,FengHKP11}, reasoning about inexact
hardware~\cite{PalemA13}, quantum complexity theory~\cite{YakaryilmazS11},
uncertainty in runtime modelling~\cite{GieseBPRIWC11}, as well as text and
speech processing~\cite{mpr02}.
PA are very expressive, as witnessed by the 
applications described above, and it is thus not surprising that most natural verification-related decision problems for them are
undecidable.  (However, equivalence and minimisation do admit
efficient algorithms~\cite{Schutzenberger61b,Tzeng92}.)
Due to these negative results, many sub-classes of probabilistic
automata have been studied. These include, for example, hierarchical~\cite{fgko15},
leaktight~\cite{csvb15} and bounded-ambiguity
automata~\cite{frw17} (see~\cite{fijalkow17} for a survey). 

We focus on the emptiness and containment problems for PA, which are natural analogs of the
like-named problems for Boolean automata.  The emptiness problem asks:
\textit{given an automaton $\Aa$, determine whether there exists a word $w$ such
that the probability $\inter{\Aa}(w)$ of it being accepted is strictly greater 
than~$\frac{1}{2}$}.\footnote{In our formulation of the emptiness problem we fix the probability threshold to be $1/2$.  Many authors allow the threshold to be any rational number between 0 and 1 and consider it as an additional input to the problem.  However it is straightforward to reduce the general case to the fixed-threshold formulation that we adopt here.}  The containment problem asks: 
\textit{given two automata~$\Aa$ and~$\Bb$,
    determine whether for all words $w$ it holds that $\inter{\Aa}(w) \leq \inter{\Bb}(w)$.} 
    The emptiness problem (also called the threshold problem) has long been known to be undecidable~\cite{paz14}.  
    It directly follows that containment, being a generalisation of emptiness, is also undecidable. 
    Undecidability of both problems holds even for automata with a fixed number of states~\cite{BlondelC03}, while in case of a unary alphabet decidability of emptiness is equivalent to longstanding open problems in number theory~\cite{AkshayAOW15}.   
    
Remarkably, even the following gap version of the emptiness problem is also undecidable~\cite{cl89}. 
The \textit{gap emptiness problem} takes as input a probabilistic automaton $\Aa$ and a rational $\epsilon \in \rat \cap (0,1)$. The task consists in distinguishing the following two cases:
\begin{enumerate}
\item there exists a word $w$ such that $\inter{\Aa}(w) > \frac{1}{2} + \epsilon$,
\item $\inter{\Aa}(w) \leq \frac{1}{2}$ for all words $w$.
\end{enumerate}
An algorithm for the gap problem is required to output "YES" in Case 1, "NO" in Case 2, and may output anything if neither case holds.
Thus the gap problem is a promise problem~\cite{Goldreich05} that can be seen as the decision analog of the problem of approximating
$\sup_{w\in\Sigma^*} \inter{\Aa}(w)$ to within given precision $\varepsilon$.   In contrast to the situation with probabilistic automata, the results of~\cite{DerksenJK05} imply that the analogous gap problem for quantum automata is decidable.

In this paper we undertake a systematic analysis of the decidability of emptiness and containment for PA in terms of their \emph{ambiguity}. 
We define the ambiguity of a PA to be that of the underlying
non-deterministic finite automaton: thus a PA is \emph{finitely ambiguous} if there exists a finite bound on the number of accepting runs of any word and \emph{polynomially ambiguous} if there exists a polynomial function 
$f:\mathbb{N} \rightarrow \mathbb{N}$ such that every word of length~$n$ has at most $f(n)$ accepting runs. 
The classes of finitely ambiguous and polynomially ambiguous automata respectively admit characterisations in terms of the structure
of their transition graphs and, based on these characterisations, it is decidable in polynomial time whether a PA is finitely ambiguous and whether it is polynomially ambiguous~\cite{WeberS91}. 

Our main results are as follows:

\paragraph{Gap emptiness problem}
We show that the gap emptiness problem, which is undecidable in general, becomes decidable when we restrict to the class of polynomially ambiguous automata.
\begin{theorem}\label{theorem:threshold}
    The gap emptiness problem is decidable for the class of polynomially
    ambiguous probabilistic automata.
\end{theorem}
To prove Theorem~\ref{theorem:threshold} we approximate polynomially ambiguous automata by finitely ambiguous automata
and rely on the decidability of emptiness in the finitely ambiguous case~\cite{frw17}.

\paragraph{Emptiness problem}
While polynomial ambiguity suffices for decidability of the gap emptiness problem, we show that the emptiness problem 
itself is undecidable even for linearly ambiguous automata.
\begin{theorem}
\label{theorem:undecidability}
    The emptiness and containment problems are undecidable for the class of linearly ambiguous
    probabilistic automata.
\end{theorem}
Theorem~\ref{theorem:undecidability} strengthens a previous result that emptiness and containment are undecidable for 
quadratically ambiguous PA~\cite{frw17}.  Recall that
emptiness is decidable for finitely ambiguous PA~\cite{frw17} and hence this result is optimal.

\paragraph{Containment problem}
Theorem~\ref{theorem:undecidability} leads us to consider the decidability 
of containment between finitely
ambiguous automata.  Here we have the following result:
\begin{theorem}
\label{theorem:decidability}
    If Schanuel's conjecture holds then 
    the containment problem 
    is decidable if at least one of the input automaton is unambiguous 
    while the other is finitely ambiguous.
\end{theorem}
We prove Theorem~\ref{theorem:decidability} by reduction to an arithmetical decision problem involving integer exponentiation,
which we solve by a process of relaxation and rounding.
The dependence on Schanuel's conjecture in this theorem is due to our use
of the result of Macintyre and Wilkie~\cite{mw96} that
the theory of real exponentiation is decidable subject to Schanuel's conjecture.
The decidability of containment between two finitely ambiguous automata remains open 
and appears to involve difficult number-theoretic considerations.

\paragraph{Organisation of the paper}
In Section~\ref{sec:definitions}, we give the formal definition of PA, the notion of ambiguity and classical results that will be useful in the paper. 
In Section~\ref{sec:problems}, we recall the problems under consideration.  In
Section~\ref{sec:threshold} we show that the gap
emptiness problem for polynomially ambiguous PA is decidable (Theorem~\ref{theorem:threshold}). Then, in Sections~\ref{sec:decidabilitykvs1} and~\ref{sec:decidability1vsk}, we prove Theorem~\ref{theorem:decidability}. In
Section~\ref{sec:decidabilitykvs1}, we explain how to translate the containment
problem into a problem about the existence of integral exponents for certain
exponential inequalities. Using this formalism, we prove that the containment
problem for $\Aa$ and $\Bb$, as stated above, is decidable if $\Aa$ is finitely
ambiguous and $\Bb$ is unambiguous.  In Section~\ref{sec:decidability1vsk}, we
tackle the more challenging direction and prove that the containment problem is
also decidable if $\Aa$ is unambiguous and $\Bb$ is finitely ambiguous.
Finally, in Section~\ref{sec:undecidability}, we prove that the emptiness and the containment
problems are undecidable provided that one of the automata is allowed to be linearly ambiguous (Theorem~\ref{theorem:undecidability}).

\section{Preliminaries}
\label{sec:definitions}
In this section, we define probabilistic automata and recall some classical
results.

\paragraph{Notation} 
We use boldface lower-case letters, e.g., $\aa, \bb, \dots$, to denote
vectors and upper-case letters, e.g., $M, N,\dots$, for matrices.
For a vector $\aa$, we write $a_i$
for its $i$-th component, and $\aa^\top$ for its transpose.

\subsection{Probabilistic automata and ambiguity} \label{sec:ambiguity}
For a finite set~$S$, we say that a function $f : S \to \rat_{\geq 0}$ is a
\emph{distribution over $S$} if $\sum_{s \in S}f(s) \le 1$.  We write
$\distr(S)$ for the set of all distributions over $S$.  We also say that a
vector $\dd = (d_1, d_2, \dots, d_n) \in \mathbb{Q}^n_{\ge 0}$ of non-negative
rationals is a distribution if~$\sum_{i=1}^n d_i \leq 1$.

A \emph{probabilistic automaton (PA)} $\Aa$ is a tuple
$(\Sigma,Q,\delta,\iota,F)$, where:
\begin{itemize}
    \item $\Sigma$ is the finite alphabet,
    \item $Q$ is the finite set of states,
    \item $\delta : Q \times \Sigma \to \distr(Q)$ is the (probabilistic)
        transition function,
    \item $\iota \in \distr(Q)$ is the initial distribution, and 
    \item $F \subseteq Q$ is the set of final states.
\end{itemize}
We write $\delta(q,a,p)$ instead of $\delta(q,a)(p)$ for the \emph{probability
of moving from $q$ to $p$ reading $a$}. Consider the word $w  = a_1 \dots a_n
\in \Sigma^*$.  A \emph{run $\rho$ of $\Aa$ over} $w = a_1 \dots a_n$ is a
sequence of transitions $(q_{0},a_1,q_{1}),(q_{1},a_2,q_{2}), \dots,
(q_{n-1},a_n,q_{n})$ where $\delta(q_{i-1},a_i,q_{i}) > 0$ for all $1 \leq i
\leq n$.  It is an \emph{accepting run} if $\iota(q_0) > 0$ and $q_n \in F$.
The \emph{probability of the run} $\rho$ is $\prob{\Aa}{\rho} \defequals
\iota(q_0) \cdot \prod_{i =1}^n \delta(q_{i-1},a_i,q_{i})$.

The automaton $\Aa$ realizes a function $\inter{\Aa}$ mapping words over the
alphabet $\Sigma$ to values in $[0,1]$. Formally, for all $w \in \Sigma^*$, we
set:
\(
    \inter{\Aa}(w) \defequals \sum_{\rho \in \accruns{\Aa}{w}} \prob{\Aa}{\rho}
\)
where $\accruns{\Aa}{w}$ is the set of all accepting runs of $\Aa$ over $w$.

\paragraph{Ambiguity}
The notion of ambiguity depends only on the structure of the underlying
automaton (i.e., whether a probability is null or not, but not on its actual
value).  An automaton $\Aa$ is said to be \textit{unambiguous} (resp.\
\textit{$k$-ambiguous}) if for all words~$w$, there is at most one accepting run
(resp.\ $k$ accepting runs) over $w$ in $\Aa$.  If an automaton is $k$-ambiguous
for some $k$, then it is said to be \textit{finitely ambiguous}.  If there
exists a polynomial~$P$, such that for every word $w$, the number of accepting
runs of $\Aa$ on $w$ is bounded by $P(|w|)$ (where $|w|$ is the length of $w$),
then $\Aa$ is said to be \textit{polynomially ambiguous}, and \textit{linearly
ambiguous} whenever the degree of~$P$ is at most~$1$.

It is well-known that if an automaton is not finitely ambiguous then it is at
least linearly ambiguous (see, for example, the criterion in~\cite[Section
3]{ws91}).  The same paper shows that if an automaton is finitely ambiguous then
it is $k$-ambiguous for $k$ bounded exponentially in the number of states of
that automaton.

We give two examples of PA and discuss their ambiguity in
Figure~\ref{fig:too-much-ambiguity}. As usual, they are depicted as graphs.  The
initial distribution is denoted by ingoing arrows associated with their
probability (when there is no such arrow, the initial probability is $0$) and
the final states are denoted by outgoing arrows.

\begin{figure}
\begin{minipage}[t]{0.35\textwidth}
\centering
\begin{tikzpicture}[scale=1]

\node[state] (q1) at (0,0) {};
\node[state] (q2) at (2.5,0) {};
\node (fant) at (-0.65,-0.65) {\small{$1$}};

\draw[->] (fant)--(q1);

\node (fant2) at (3,-0.6) {};
\draw[->] (q2)--(fant2);

\path[trans]
(q1) edge[loop above] node {\small{$a : \frac{1}{2}$}} (q1)
(q1) edge node[above] {\small{$b : 1$}} (q2)
(q2) edge[loop above] node {\small{$a,b : 1$}} (q2)
;
\end{tikzpicture}
\end{minipage}
\hfill
\begin{minipage}[t]{0.6\textwidth}
\centering
\begin{tikzpicture}[scale=1]

\node[state] (q1) at (0,0) {};
\node[state] (q2) at (2.5,0) {};
\node[state] (qt) at (-2.5,0)  {\small{$q_\bot$}};

\path[trans]
(q1) edge[loop above] node {\small{$a : \frac{1}{2}$}} (q1)
(q1) edge node[above] {\small{$b : 1$}} (q2)
(q1) edge node[above] {\small{$a : \frac{1}{2}$}} (qt)
(q2) edge[loop above] node {\small{$a,b : 1$}} (q2)
(qt) edge[loop above] node {\small{$a,b : 1$}} (qt)
;

\node (fant) at (-0.65,-0.65) {\small{$1$}};
\draw[->] (fant)--(q1);

\node (fant2) at (0.6,-0.6) {};
\draw[->] (q1)--(fant2);

\node (fant3) at (-1.9,-0.6) {};
\draw[->] (qt)--(fant3);
\end{tikzpicture}
\end{minipage}
\caption{Two PA over the alphabet $\Sigma = \{a, b\}$ are
depicted.  On the left hand side, automaton~$\Aa$ induces the function
$a^nb\Sigma^* \mapsto \frac{1}{2^n}$ and $a^* \mapsto 0$.  On the right hand
side, the automaton~$\overline{\Aa}$ induces the function $a^nb\Sigma^* \mapsto
1 - \frac{1}{2^n}$ and $a^* \mapsto 1$.  Observe that $\Aa$ is unambiguous
and~$\overline{\Aa}$ is linearly ambiguous.}
\label{fig:too-much-ambiguity}
\end{figure}
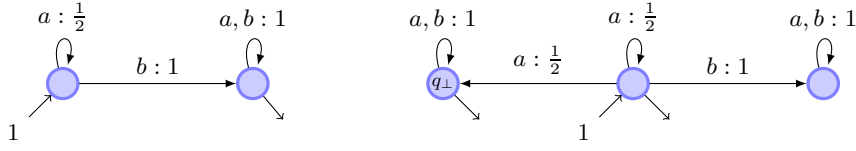

\subsection{Classical results}
\paragraph{Weighted-sum automaton} For PA~$\Aa_1, \Aa_2,
\dots, \Aa_n$ over the same alphabet, and for a discrete distribution $\dd =
(d_1, d_2, \dots, d_n)$, the \emph{weighted-sum automaton} (of~$\Aa_1, \Aa_2, \dots,
\Aa_n$ with weights~$\dd$) is defined to be the disjoint union of the $n$
automata with the initial distribution $\iota(q) \defequals d_i \cdot
\iota_i(q)$ if $q$ is a state of~$\Aa_i$, where $\iota_i$ is the initial
distribution of~$\Aa_i$.  Note that if $\Bb$ is the weighted sum of~$\Aa_1,
\Aa_2, \dots, \Aa_n$ with weights~$\dd$ then it is also a probabilistic
automaton and $\inter{\Bb}  = \sum_{i=1}^n d_i \cdot \inter{\Aa_i}$.

\paragraph{Complement automaton} For a PA~$\Aa$, we
define its \emph{complement automaton} $\overline{\Aa}$ in the following way.
First, define the PA~$\Aa'$ by modifying~$\Aa$ as follows:
\begin{itemize}
    \item add a new sink state~$q_\bot$;
    \item obtain the transition function~$\delta'$ from~$\delta$ by adding
        transitions:
        \begin{itemize}
            \item  $\delta'(q_\bot, a, q_\bot) = 1$ for all $a \in \Sigma$,
            \item $\delta'(q, a, q_\bot) = 1 - \sum_{r \in Q} \delta(q, a, r)$
                for all $(q, a) \in Q \times \Sigma$;
        \end{itemize}
    \item obtain the initial distribution~$\iota'$ from~$\iota$ by adding
        $\iota'(q_\bot) = 1 - \sum_{q \in Q} \iota(q)$.
\end{itemize}
Observe that $\inter{\Aa'} = \inter{\Aa}$, that $\sum_{r \in Q} \delta'(q, a, r)
= 1$ for all $(q, a) \in Q \times \Sigma$, and that $\sum_{q \in Q} \iota'(q) =
1$.  We obtain~$\overline{\Aa}$ from~$\Aa'$ by swapping its final and non-final
states.  As expected, it is the case that $\inter{\overline{\Aa}} = 1 -
\inter{\Aa}$.

\begin{remark}[Preserving ambiguity]
    The ambiguity of a weighted-sum automaton is at most the sum of 
    the ambiguities of the individual automata,
    and the ambiguity of a complement automaton may be larger than the 
    ambiguity of the original one
    (see Figure~\ref{fig:too-much-ambiguity}). 
\end{remark}

\section{Decision problems}\label{sec:problems}
In this work, we are interested in comparing the functions computed by
PA.  We write $\inter{\Aa} \leq \inter{\Bb}$ if ``$\Aa$ is
contained in $\Bb$'', that is if $\inter{\Aa}(w) \leq \inter{\Bb}(w)$ for all $w
\in \Sigma^*$; and we write~$\inter{\Aa} \leq \frac{1}{2}$ if $\inter{\Aa}(w) \leq
\frac{1}{2}$ for all $w \in \Sigma^*$.  We are interested in the following
decision problems for PA.
\begin{itemize}
    \item \textit{Containment problem:} Given probabilistic automata~$\Aa$
        and~$\Bb$, does $\inter{\Aa} \leq \inter{\Bb}$ hold?
    \item \textit{Emptiness problem:} Given a probabilistic automaton $\Aa$,
        does $\inter{\Aa} \leq \frac{1}{2}$ hold?
        \item \textit{Gap emptiness problem:}
\begin{itemize}
\item \textbf{Input}: $\epsilon \in \rat \cap (0,1)$ and a probabilistic automaton $\Aa$
such that either there is a word $w$ satisfying $\inter{\Aa}(w) > \frac{1}{2} + \epsilon$ or $\inter{\Aa}(w)
\leq \frac{1}{2}$ for all words $w$.
\item \textbf{Output}: does $\inter{\Aa}(w) \leq \frac{1}{2}$ hold?
\end{itemize}
\end{itemize}

We will argue that the containment and emptiness problems are both undecidable
when considered for the class of linearly ambiguous automata
(Section~\ref{sec:undecidability}). 
The emptiness problem is known to be
decidable for the class of finitely ambiguous automata~\cite{frw17}. We tackle
here the more difficult containment problem (Sections~\ref{sec:decidabilitykvs1}
and~\ref{sec:decidability1vsk}).

As for the gap emptiness problem, recall that for general PA, it is
known to be undecidable~\cite{cl89}.  We know that it is decidable for
finitely-ambiguous PA (because the emptiness problem is decidable).
Section~\ref{sec:threshold} is devoted to clarifying the decidability frontier
for polynomially-ambiguous PA.

\section{Decidability of gap emptiness for polynomially ambiguous automata}
\label{sec:threshold}

This section is devoted to proving that the gap emptiness problem is
decidable for the class of polynomially ambiguous PA.

For the rest of the section, fix a rational $\epsilon \in (0,1)$ and a PA $\Aa =
(\Sigma,Q,\delta,\iota,F)$ that is polynomially ambiguous. We also assume,
without loss of generality, that $\Aa$ is trimmed (i.e., all states are reachable
from some initial
state and can reach some final state).

The key ingredient of this section is to show that we can compute a PA $\Aa'$ such that:
\begin{itemize}
\item $\Aa'$ is finitely ambiguous,
\item for all words $w$, $\inter{\Aa'}(w) \leq \inter{\Aa}(w) \leq \inter{\Aa'}(w) + \epsilon$. 
\end{itemize} 

Using such a construction we can easily prove Theorem~\ref{theorem:threshold} reducing the question to the emptiness problem of $\Aa'$ (which is decidable since $\Aa'$ is finitely ambiguous). 
Indeed, suppose that for all $w$ we have $\Aa'(w) \leq \frac{1}{2}$. Then for all $w$, $\Aa(w) \leq \frac{1}{2} + \epsilon$. Hence there does not exist a word $w$ such that $\Aa(w) > \frac{1}{2} + \epsilon$.
Conversely, if there exists $w$ such that $\Aa'(w) > \frac{1}{2}$ then $\Aa(w) > \frac{1}{2}$. But then it is not the case that for all $w$ we have $\Aa(w) \leq \frac{1}{2}$.

\paragraph*{Construction of $\Aa'$}
Let $N$ be a positive integer ($N$ will be fixed later in the proof, depending
only on $\epsilon$ and $\Aa$). Let $\Aa'$ be the same as $\Aa$ except that on
every run, we are only allowed to make the first $N$ non-deterministic choices
(including the choice of the initial state). In other words, we can only take at
most $N$ times a transition $(p,a,q)$ of non-zero probability such that there
exists another transition $(p,a,q')$ of non-zero probability. After seeing the
$(N+1)$-th non-deterministic choice, the automaton rejects the run. This can be
achieved easily by making $N$ copies of $\Aa$. Clearly, $\Aa'$ is
finitely ambiguous: for every word, there are at most $|Q|^N$ accepting runs.
Moreover, since the runs of $\Aa'$ can be embedded in the runs of $\Aa$,
$\inter{\Aa'}(w) \leq \inter{\Aa}(w)$ for all words. It remains to prove that $\inter{\Aa}(w) \leq \inter{\Aa'}(w) + \epsilon$ for all words $w$.

\paragraph*{Proof that $\inter{\Aa}(w) \leq \inter{\Aa'}(w) + \epsilon$ for all words $w$}

The combination of the two following lemmas will prove $\inter{\Aa}(w) \leq \inter{\Aa'}(w) + \epsilon$ for all words $w$.

\begin{lemma}
\label{lemma:weights}
Any run in $\Aa$ using exactly $m$ non-deterministic choices has probability at
most $\alpha^m$, where $\alpha$ is the maximal transition probability not equal to $1$
in~$\Aa$.
\end{lemma}

\begin{lemma}
\label{lemma:numberofruns}
There exists a polynomial $P_\Aa$ such that for all words $w$, there are at most
$P_{\Aa}(m)$ accepting runs of $\Aa$ over $w$ using exactly $m$ non-deterministic choices.
\end{lemma}

Let us first prove that these two lemmas lead to the result.  By definition of~$\Aa'$, 
Lemma~\ref{lemma:weights}, and Lemma~\ref{lemma:numberofruns}, we have:
\begin{align*}
\inter{\Aa}(w) - \inter{\Aa'}(w) \le \sum_{m = N+1}^\infty \alpha^{m} P_{\Aa}(m) 
\end{align*}
for some polynomial function $P_{\Aa}$, where $\alpha$ is the maximal transition
probability not equal to $1$ in $\Aa$.

Since $\alpha < 1$, it is easily verified that the series $\sum_{m = 0}^\infty
\alpha^{m} P_{\Aa}(m)$ converges (e.g. by the d'Alembert's ratio test). This is
equivalent to 
\[
    \lim_{N \to \infty} \sum_{m = N+1}^\infty \alpha^{m} P_{\Aa}(m) = 0.
\]
Hence it suffices to take $N$ such that $\sum_{m = N+1}^\infty \alpha^{m}
P_{\Aa}(m) \leq \epsilon$. Such an $N$ is computable from $\epsilon$ and $\Aa$.
Hence, it suffices to prove Lemma~\ref{lemma:weights} and
Lemma~\ref{lemma:numberofruns}. 

Lemma~\ref{lemma:weights} is immediate: by
definition, on a run using $m$ non-deterministic choices, there are at least $m$
transitions with probability smaller than $1$, and thus with probability at most
$\alpha$ (while the other transitions have probability at most $1$).

Let us now turn to the proof of Lemma~\ref{lemma:numberofruns}. 

\begin{lemma}\label{lemma:polynomial}
For all words $w$, there are at most $2^{|Q|}((m+1)|Q|^2)^{|Q|^{3}}$ runs of
$\Aa$ over $w$ using exactly $m$ non-deterministic choices.
\end{lemma}

The rest of this section is devoted to proving this lemma, which will conclude the proof.

We define the strongly connected components (SCC) of $\Aa$ as the SCCs of the
underlying graphs when ignoring the labels of the transitions (only considering
transitions with positive probability).

\begin{lemma}
For all states $p,q$ in a same SCC of $\Aa$ and for all words $w$, there is at
most one run from $p$ to $q$ on $w$.
\end{lemma}

Note that by definition of SCC, any such run remains in the SCC.

\begin{proof}
We use a criterion from~\cite{ws91} that characterizes ambiguity of automata.
%
The EDA criterion states that an automaton is not polynomially ambiguous if and only if there is a word $w \in \Sigma^*$ and a state $s \in Q$ such that there are at least two different runs from $s$ to $s$ on $w$.


Suppose now that there is a word $w$ such that there are two distinct runs from $p$ to $q$ as defined in the lemma. Since $p$ and $q$ belong to the same SCC, there is a word $w'$ and a run from $q$ to $p$ on $w'$. Thus there are two distinct runs from $p$ to $p$ on $ww'$. Using the EDA criterion, $\Aa$ cannot be polynomially ambiguous and we get a contradiction. 
\end{proof}

As a direct consequence, we have:

\begin{corollary}\label{cor:SCCruns}
For all words $w$, there are at most $|Q|^2$ runs on $w$ that start and end in
any given SCC.
\end{corollary}

The next lemma is the last ingredient needed to prove Lemma~\ref{lemma:polynomial}.

\begin{lemma}\label{lemma:SCCruns} 
For all words $w$, the sum of all runs of $\Aa$ over all prefixes of $w$ starting
and ending in a given SCC and using at most $m$ non-deterministic choices is bounded by $((m+1)|Q|^2)^{|Q|^2}$.
\end{lemma}

\begin{proof}
The lemma can be rephrased in terms of counting the number of leaves in a tree.
More precisely: runs correspond to branches (which are artificially rooted
together); internal nodes to states in runs in the given SCC; and leaves
correspond to states outside of the given SCC (so when a run leaves the given
SCC) or states after reading the whole $w$. Non-deterministic choices correspond to nodes having several children.

We consider finite trees. The width of the tree is the maximal number of
non-leaf nodes at a same depth (the root has depth $0$) and we say that a node
is a \textit{split} if it has more than one child. Recall that by definition of
$\Aa$ the outdegree of every node is bounded by $|Q|$. By
Corollary~\ref{cor:SCCruns}, the lemma boils down to proving that the number of leaves in a tree which: (1) has width at most $|Q|^2$, (2) has outdegree at most $|Q|$, (3) has at most $m$ split nodes on each branch, is at most $((m+1)|Q|^2)^{|Q|^2}$. Notice that the bound on the outdegree does not follow from the width bound because children can be leaves. Figure~\ref{fig:treewidth} visualises the runs on a SCC as trees on some example automaton.

\begin{figure}
    \centering
        \begin{tikzpicture}
        \path
        (0,1) edge[-,line width=5pt,black!20] ++ (8,0)
        (0,2.5) edge[-,line width=5pt,black!20] ++ (8,0)
        (0,4) edge[-,line width=5pt,black!20] ++ (8,0)
        ;
        
        \node at (1,6) {$a$};
        \node at (3,6) {$a$};
        \node at (5,6) {$b$};
        \node at (7,6) {$a$};
        
        \node[circle,minimum size=4pt,inner sep=0pt,draw=black,fill=black] (r) at (-2, 2.5) {};
        
        \node[circle,minimum size=4pt,inner sep=0pt,draw=black,fill=black,label=below:{$q_1$}] (r1) at (0, 1) {};
        \node[circle,minimum size=4pt,inner sep=0pt,draw=black,fill=black,label=below:{$q_2$}] (r2) at (0, 2.5) {};
        \node[circle,minimum size=4pt,inner sep=0pt,draw=black,fill=black,label=below:{$q_3$}] (r3) at (0, 4) {};
        
        \path
        (r) edge[-,line width=1pt] (r1)
        (r) edge[-,line width=1pt] (r2)
        (r) edge[-,line width=1pt] (r3)
        ;
        
        \node[circle,minimum size=4pt,inner sep=0pt,draw=black,fill=black,label=below:{$q_3$}] (r11) at (2, 1) {};
        \node[circle,minimum size=4pt,inner sep=0pt,draw=black,fill=black,label=below:{$q_2$}] (r22) at (2, 2.5) {};
        \node[circle,minimum size=4pt,inner sep=0pt,draw=black,fill=black,label=below:{$q_1$}] (r33) at (2, 4) {};
        
        \node[circle,minimum size=4pt,inner sep=0pt,draw=red,fill=red] (p1) at (2, 4.75) {};
        \node[circle,minimum size=4pt,inner sep=0pt,draw=red,fill=red] (p2) at (2, 3.25) {};
        
        \path
        (r1) edge [-,line width=1pt] (r11)
        (r1) edge [-,line width=1pt] (r22)
        (r3) edge [-,line width=1pt] (r33)
        (r3) edge [-,line width=1pt] (p1)
        (r2) edge [-,line width=1pt] (p2)
        ;
        
        \node[circle,minimum size=4pt,inner sep=0pt,draw=black,fill=black,label=below:{$q_1$}] (r111) at (4, 1) {};
        \node[circle,minimum size=4pt,inner sep=0pt,draw=black,fill=black,label=below:{$q_2$}] (r222) at (4, 2.5) {};
        \node[circle,minimum size=4pt,inner sep=0pt,draw=black,fill=black,label=below:{$q_3$}] (r333) at (4, 4) {};
        
        \node[circle,minimum size=4pt,inner sep=0pt,draw=red,fill=red] (p11) at (4, 1.75) {};
        \node[circle,minimum size=4pt,inner sep=0pt,draw=red,fill=red] (p22) at (4, 0.25) {};
        
        \path
        (r11) edge [-,line width=1pt] (r111)
        (r33) edge [-,line width=1pt] (r222)
        (r33) edge [-,line width=1pt] (r333)
        (r22) edge [-,line width=1pt] (p11)
        (r11) edge [-,line width=1pt] (p22)
        ;
        
        \node[circle,minimum size=4pt,inner sep=0pt,draw=black,fill=black,label=below:{$q_1$}] (r2222) at (6, 2.5) {};
        \node[circle,minimum size=4pt,inner sep=0pt,draw=black,fill=black,label=below:{$q_3$}] (r3333) at (6, 4) {};
        
        \node[circle,minimum size=4pt,inner sep=0pt,draw=red,fill=red] (p111) at (6, 4.75) {};
        \node[circle,minimum size=4pt,inner sep=0pt,draw=red,fill=red] (p222) at (6, 1.75) {};
        
        \path
        (r222) edge [-,line width=1pt] (r2222)
        (r222) edge [-,line width=1pt] (r3333)
        (r111) edge [-,line width=1pt] (p222)
        (r333) edge [-,line width=1pt] (p111)
        ;
        
        \node[circle,minimum size=4pt,inner sep=0pt,draw=red,fill=red,label=below:{$q_2$}] (r11111) at (8, 1) {};
        \node[circle,minimum size=4pt,inner sep=0pt,draw=red,fill=red,label=below:{$q_3$}] (r22222) at (8, 2.5) {};
        \node[circle,minimum size=4pt,inner sep=0pt,draw=red,fill=red,label=below:{$q_1$}] (r33333) at (8, 4) {};
        
        \node[circle,minimum size=4pt,inner sep=0pt,draw=red,fill=red] (p1111) at (8, 3.25) {};
        
        \path
        (r2222) edge [-,line width=1pt] (r11111)
        (r2222) edge [-,line width=1pt] (r22222)
        (r3333) edge [-,line width=1pt] (r33333)
        (r3333) edge [-,line width=1pt] (p1111)
        ;
        \end{tikzpicture}
\caption{In this example the set of states is $Q = Q_1 \cup P$, where $Q_1 = \{q_1, q_2, q_3\}$ are states in the given SCC and $P$ are the remaining states.
We consider runs over the word $w = aaba$.
The transitions are depicted only on the picture.
The width is $3$ and this bound is highlighted with grey shadows. The inner nodes are coloured black and the leaves are coloured red. The unlabelled leaves have states from $P$.}\label{fig:treewidth}
\end{figure}
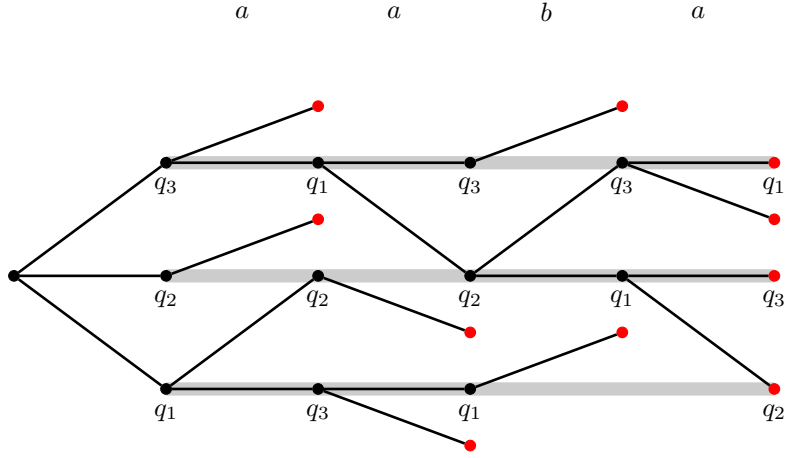

We will show a more general claim, which also proves Lemma~\ref{lemma:SCCruns}.

\begin{claim}
    For all non-negative integers $k$ and $m$, the number of leaves in a tree which
    has
    \begin{enumerate}
        \item[(1)] width and outdegree at most $k$ and
        \item[(2)] at most $m$ split nodes on each branch
    \end{enumerate}
    is at most $((m+1)k)^k$.
\end{claim}
The proof of the claim is by induction on $m$ and $k$. The base case $m=0$ is
immediate, since a tree with no split is a path and has one leaf. Moreover the
case $k=0$ corresponds to a tree consisting of only one node and the induction
hypothesis is also satisfied. Now, consider a tree with positive parameters $m$
and~$k$. We will focus on the first split node, that is
the node where the single branch from the root splits for
the first time. Let us call it $v$. By assumption, $v$ has at most $k$ children.
Furthermore, it is not hard to see that there is
at most one subtree of $v$ with width $k$. Hence,
the number of leaves in the tree is bounded by the sum of the leaves
of
\begin{itemize}
    \item one subtree of $v$ with width $k$ and at most $m-1$ split notes, and
    \item $k$ subtrees of $v$ with width at most $k-1$ and at most $m-1$ split nodes.
\end{itemize}
Overall, by the induction
hypothesis, the number of leaves in the tree is thus bounded by $(mk)^k +
k(m(k-1))^{k-1}$ which is smaller than $((m+1)k)^k$.
\end{proof}

\begin{proof}[Proof of Lemma~\ref{lemma:polynomial}]
We prove the lemma by induction on the number of states in $\Aa$. The case for one state is immediate. Let $w$ be a word and consider now a source-SCC in $\Aa$ that is to say an SCC such that there is no transition from any other SCC to this one. Such a source-SCC always exists.
By Lemma~\ref{lemma:SCCruns}, there are at most $((m+1)|Q|^2)^{|Q|^2}$ runs on a
prefix of $w$ in the source-SCC, using at most $m$ non-deterministic choices. 
If any run is continued after leaving this source-SCC we can bound the number of runs that arise from it by the induction hypothesis because the number of states has decreased.
Since every run either starts in this source-SCC (and if it leaves it, it does not go back to it) or never reaches it,
by induction hypothesis, there are at most
\begin{multline*}
    ((m+1)|Q|^2)^{|Q|^2} 2^{|Q|-1}((m+1)(|Q|-1)^2)^{(|Q|-1)^{3}} \\
    + 2^{|Q|-1}((m+1)(|Q|-1)^2)^{(|Q|-1)^{3}}
\end{multline*}
runs on $w$ in $\Aa$ using at most $m$ non-deterministic choices.
This concludes the proof since:
\begin{multline*}
  ((m+1)|Q|^2)^{|Q|^2} 2^{|Q|-1}((m+1)(|Q|-1)^2)^{(|Q|-1)^{3}} \\
     \shoveright{+ 2^{|Q|-1}((m+1)(|Q|-1)^2)^{(|Q|-1)^{3}}} \\
     \begin{aligned}
       & \leq \quad 2^{|Q|-1}((m+1)|Q|^2)^{|Q|^2 (1 + |Q| - 1)} + 2^{|Q|-1}((m+1)|Q|^2)^{|Q|^{3}} \\
       & = \quad 2^{|Q|}((m+1)|Q|^2)^{|Q|^{3}}
     \end{aligned}
\end{multline*}
\end{proof}

\section{Decidability of containment: the case finitely ambiguous vs. unambiguous}
\label{sec:decidabilitykvs1}
Our aim is to decide whether $\inter{\Aa} \leq \inter{\Bb}$. 
We first give a 
translation of the problem into a problem about 
the existence of integral exponents
for certain exponential inequalities. 

\paragraph{Notation} In the rest of the paper, we write $\exp(x)$ to denote
the exponential function $x \mapsto e^x$, and $\log(y)$ for the natural
logarithm function $y \mapsto \log_e(y)$. For a real number $x$ and a positive
real number $y$, we write $y^x$ for $\exp(x\log(y))$.

\subsection{Translating the containment problem into exponential inequalities}
\label{sec:equivalent}
We are going to translate the negation of the containment problem: Given two
finitely ambiguous PA $\Aa$ and $\Bb$, does there exist a
word $w$, such that $\inter{\Aa}(w) > \inter{\Bb}(w)$?  Consider two positive
integers $k$ and $n$, and vectors $\pp \in \mathbb{Q}_{> 0}^k$ and $\qq_1,
\dots, \qq_k \in \mathbb{Q}_{> 0}^n$.
We denote by $S(\pp,\qq_1,\dots,\qq_k): \mathbb{N}^n \to \mathbb{R}$ the
function associating a vector $\xx \in \mathbb{N}^n$ to
\(
	\sum_{i=1}^{k}p_i q_{i,1}^{x_1} \ldots q_{i,n}^{x_n},
\)
where $q_{i,j}$ is the $j$-th component of vector $\qq_i$.

\begin{proposition}
\label{proposition:translation}    
    Given a $k$-ambiguous automaton~$\Aa$ and an $\ell$-ambiguous
    automaton~$\Bb$, one can compute a positive integer $n$ and a finite set
    $\Delta$ of tuples $(\pp,\qq_1,\dots,\qq_{k'},\rr,\ss_1,\dots,\ss_{\ell'})$
    of vectors $\pp \in \mathbb{Q}_{>0}^{k'},\rr \in \mathbb{Q}_{>0}^{\ell'}$,
    for some $k' \leq k$ and $\ell' \leq \ell$; and $\qq_{i} \in
    \mathbb{Q}_{>0}^n, \ss_{j} \in \mathbb{Q}_{>0}^n$, for all $i$ and $j$; such
    that the following two conditions are equivalent:
    \begin{itemize}
        \item there exists $w \in \Sigma^*$ such that $\inter{\Aa}(w) >
            \inter{\Bb}(w)$,
        \item there exist $(\pp,\qq_1,\dots,\qq_{k'},
            \rr,\ss_1,\dots,\ss_{\ell'}) \in \Delta$ and $\xx \in \mathbb{N}^n$
            such that    
            \[
                S(\pp,\qq_{1},\dots,\qq_{k'})(\xx) >
                S(\rr,\ss_{1},\dots,\ss_{\ell'})(\xx).
            \]
    \end{itemize}
\end{proposition}
It thus follows that to prove Theorem~\ref{theorem:decidability}, it suffices to
show decidability of the second item of
Proposition~\ref{proposition:translation} for a given element of $\Delta$ in the
cases where either $k$ or $\ell$ are equal to $1$.
The proof of Proposition~\ref{proposition:translation} is technical and is postponed to Section~\ref{sec:translation}.


\begin{example}
    Consider the following instance of the problem, where $k = n = 2$, $\ell =
    1$, and $p$ is a fixed rational number $0 \le p \le 1$:  Do there exist $x,y
    \in \nat$ such that $p\cdot \left(\frac{1}{12} \right)^x \cdot
    \left(\frac{1}{2}\right)^y + (1-p)\cdot \left(\frac{1}{3} \right)^x \cdot
    \left(\frac{1}{18}\right)^y < \left(\frac{1}{6} \right)^x \cdot
    \left(\frac{1}{6}\right)^y$. This can be rewritten as
    \[
        p\cdot \left(\frac{1}{2} \right)^x \cdot 3^y + (1-p)\cdot 2^x \cdot
        \left(\frac{1}{3}\right)^y < 1
    \]
    or equivalently, using the exponential function, as follows
    \[
        \exp(\log(p) - x\log(2) + y\log(3)) + \exp(\log(1-p) + x\log(2) -
        y\log(3)) < 1.
    \]
    Consider the set $V = \{(x,y) \in \real^2 \mid e^x + e^y < 1\}$ and denote
    by $b$ the point $(\log(p), \log(1-p))$. Let $\uu = (-\log(2), \log(2))$ and
    $\vv = (\log(3), -\log(3))$ be two vectors. See Figure~\ref{fig:example} for
    a geometric representation.  The question is now: do there exist $x,y \in
    \nat$ such that $b + x\uu + y\vv \in V$. We will show that the answer
    is yes if and only if $p \neq \frac{1}{2}$.
    
    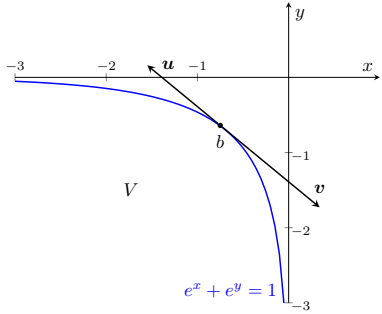
\begin{figure}
    \centering
        \begin{tikzpicture}[scale=0.7]
        \begin{axis}[
                  xmax=1,ymax=1,
                  xmin=-3,ymin=-3,
                  axis lines=middle,
                  xlabel={$x$},
              ylabel={$y$},
              xtick={-3,-2,-1},
              ytick={-3,-2,-1},
              yticklabel style = {font=\scriptsize,xshift=4ex,yshift=-0.5ex},
              xticklabel style = {font=\scriptsize,yshift=3.3ex}
        ]
        \addplot[blue,domain=-5:0,samples=100,thick]
            {(log2(1 - pow(2.7,x))/1.44)} node[above left] {$e^x + e^y = 1$};
        \node[label={180:{$V$}}] at (axis cs:-1.5,-1.5) {};
        \node[label={270:{$b$}},circle,fill,inner sep=1pt] (b)
            at (axis cs:-0.75,-0.64) {};

        \coordinate (u) at (axis direction cs:-0.8,0.8);
        \coordinate (v) at (axis direction cs:1.09,-1.09);

        \path[->]
        ($(b)+(u)$) edge[<-,thick,>=stealth] (b) node[right=1ex] {$\uu$};
        \path[->]
        ($(b)+(v)$) edge[<-,thick,>=stealth] (b) node[above=1ex] {$\vv$};

        \end{axis}
        \end{tikzpicture}
    \caption{The set $V$ is bounded by the plot $e^x + e^y = 1$ and the point
    $b$ is on that plot.}
    \label{fig:example}
    \end{figure}

    Let $C = \{(x,-x) \mid x \in \real\}$.  For $p=\frac{1}{2}$, the affine line
    $C + p$ is tangent to the blue plot and so, whatever the values of $x$ and
    $y$, $b + x\uu + y\vv$ cannot be in $V$.  For $p \neq \frac{1}{2}$, there is
    a value $\delta$ such that the whole interval strictly between $b$ and $b +
    (\delta,-\delta)$ is in $V$. Since $\log(2)$ and $\log(3)$ are rationally
    independent, the set $D = \{x\uu + y\vv \mid x,y \in \nat\}$ is a dense
    subset of~$C$, so in particular, there is a point of $D + b$ in the interval
    between $b$ and $b + (\delta,-\delta)$ and thus there exist $x,y \in \nat$
    such that $b + x\uu + y\vv \in V$.
\end{example}

\subsection{Decidability: the case finitely ambiguous vs. ambiguous}
We prove here the decidability of the containment problem 
when $\Aa$ is finitely ambiguous and 
$\Bb$ is unambiguous. The converse situation is tackled 
in Section~\ref{sec:decidability1vsk}.
\begin{proposition}
	Determining whether $\inter{\Aa} \leq \inter{\Bb}$
    is decidable when $\Aa$
    is finitely ambiguous and $\Bb$ is unambiguous.
\end{proposition}
\begin{proof}
    Let $\Aa$ be $k$-ambiguous.  Proposition~\ref{proposition:translation} shows
    that it is sufficient to decide, given an integer $n$ and positive rational
    numbers $p, q_j, r_i, s_{i,j}$ for $i\in\{1,\ldots,k\}$,
    $j\in\{1,\ldots,n\}$, whether there exists $x_1,\dots,x_n \in \nat$ such
    that
    \begin{equation} \label{eq:kvs1}
        \sum_{i=1}^k p_i q_{i,1}^{x_1} \cdots q_{i,n}^{x_n} > r
        s_{1}^{x_1} \cdots s_{n}^{x_n}.
    \end{equation}

    We consider two cases. First, assume that there exist $i$ and $j$ such that
    $q_{i,j} > s_j$. Then in that case, for a large enough $m \in \mathbb{N}$
    condition~\eqref{eq:kvs1} will be satisfied for $(x_1,\ldots,
    x_j,\ldots,x_n) = (0, \ldots, m, \ldots, 0)$. Otherwise, assume that $\max\{
    q_{i,j} \st 1 \le i \le k \} \leq s_j$ for all $1 \le j \le n$. In this
    case, if there exists a valuation of the $x_i$
    satisfying~\eqref{eq:kvs1} then $(x_1, \dots, x_n) = (0, \dots, 0)$ also
    satisfies it. It is then sufficient to test condition~\eqref{eq:kvs1}
    for $x_1 = \dots = x_n = 0$ to conclude.
\end{proof}

\subsection{Proof of Proposition~\ref{proposition:translation}}
\label{sec:translation}
We write that the first state of a run is the first state of the first transition, and the last state of a run is the last state of the last transition.
A run is \emph{simple} if every state appearing in the sequence of transitions composing the run, appears at most twice.
A \emph{cycle} is a run in which the first and the last states coincide. 
A \emph{simple cycle} is a cycle which is a simple run.

For a state~$q$, we say that a cycle is a \emph{$q$-cycle} if the first
(and hence also the last) state is~$q$.
For a run $\rho = \rho' \cdot \rho''$, such that the last state of~$\rho'$ 
(and hence also the first state of~$\rho''$) is~$q$, and for a $q$-cycle
$\omega$, the result of \emph{injecting $\omega$ (after~$\rho'$) into~$\rho$} is
the run~$\rho' \cdot \omega \cdot \rho''$. 

For a run~$\rho$, we write $Q(\rho)$ for the set of states that occur in it.
For a set of states~$P$, we write $\periods(P)$ for the set of simple cycles
in which only states in~$P$ occur.
A \emph{simple cycle decomposition} is a pair $(\gamma, \sigma)$, where 
$\gamma$ is a run of length less than~$|Q(\gamma)|^2$ and 
$\sigma : \periods(Q(\gamma)) \to \nat$.  
We say that a simple cycle decomposition $(\gamma, \sigma)$ is a 
\emph{simple cycle decomposition of a run~$\rho$} if the run~$\rho$ can be 
obtained from~$\gamma$ by injecting $\sigma(\omega)$ cycles~$\omega$, for 
every simple cycle $\omega \in \periods(Q(\gamma))$, in some order.

\begin{proposition}
\label{prop:simple_cycle}
    Every run has a simple cycle decomposition.
\end{proposition}

\begin{proof}
The above result is classic, cf.\ e.g.~\cite[proof of Lemma~4.5]{Rackoff78}.  It follows by repeatedly removing from a run~$\rho$, as long as its length is at least~$|Q(\rho)|^2$, some simple cycle whose removal does not decrease the set of states~$Q(\rho)$.  To see that such a simple cycle must exist, observe that if the length of~$\rho$ is at least~$|Q(\rho)|^2$, then it contains $|Q(\rho)|$ non-overlapping simple cycles~$\omega_1$,~\ldots, $\omega_{|Q(\rho)|}$; and if~$W_i$ is the set of all states that occur strictly inside~$\omega_i$ but nowhere else in~$\rho$, then the sets~$W_1$,~\ldots, $W_{|Q(\rho)|}$ are mutually disjoint and their union has size less than~$|Q(\rho)|$, so some~$W_i$ must be empty.
\end{proof}

If $(\gamma, \sigma)$ is a simple cycle decomposition, then we refer to 
$\gamma$ as its \emph{spine} and to $\sigma$ as its 
\emph{simple cycle count}. 
Observe that the number of distinct spines is finite;
more specifically, it is at most exponential in the size of the automaton, 
as is the set of simple cycles~$\periods(Q(\gamma))$ for every 
spine~$\gamma$. 

We say that a simple cycle decomposition $(\gamma, \sigma)$ is 
\emph{accepting} if the run~$\gamma$ is.
By Proposition~\ref{prop:simple_cycle}, every accepting
run has an accepting simple cycle decomposition.
Moreover, for every accepting spine~$\gamma$, and for every function
$\sigma : \periods(Q(\gamma)) \to \nat$, there is at least one accepting 
run~$\rho$, such that $(\gamma, \sigma)$ is its simple cycle decomposition.

\begin{proposition}
\label{proposition:fixed-ambiguity}
  There is an algorithm that given a finitely ambiguous probabilistic 
  automaton~$\Aa$ and a nonnegative integer~$i$, outputs a finite automaton
  that accepts the language of words on which~$\Aa$ has exactly~$i$ accepting
  runs.
\end{proposition}

\begin{proof}
  We can assume that $\Aa$ is trimmed (i.e., all states are reachable 
  from some initial state and can reach some final state).
  It is known that the number of all active runs in a trimmed finite ambiguous automaton is bounded exponentially in the number of states in $\Aa$~\cite{ws91}.
  We can therefore add an extra component to $\Aa$ that, using the powerset construction,  
  keeps track of all active runs in the automaton. 
  Using this component, the automaton $\Aa$ can extract the number 
  of all accepting runs.
\end{proof}


We are ready to prove Proposition~\ref{proposition:translation}.

\begin{proof}[Proof of Proposition~\ref{proposition:translation}]
  First, use Proposition~\ref{proposition:fixed-ambiguity} to compute finite
  automata $\Aa_{k'}$, $0 \leq k' \leq k$, and $\Bb_{\ell'}$, 
  $0 \leq \ell' \leq \ell$, that accept the languages of words on 
  which~$\Aa$ has exactly~$k'$ accepting runs and $\Bb$ has exactly~$\ell'$ 
  accepting runs, respectively.
  
  For all~$k'$, $0 \leq k' \leq k$, and for all $\ell'$, 
  $0 \leq \ell' \leq \ell$, we perform the following.
  Consider the synchronized product of~$\Aa_{k'}$, $\Bb_{\ell'}$, 
  $k'$~copies of~$\Aa$, and $\ell'$~copies of~$\Bb$.
  Moreover, equip the synchronized product with another component, 
  a finite automaton that maintains (in its state space) the partition of 
  the $k'$ components corresponding to copies of~$\Aa$, and of the partition 
  of the $\ell'$ components corresponding to copies of~$\Bb$, that reflects 
  which of the $k'$ runs of~$\Aa$, and which of the $\ell'$ runs of~$\Bb$, 
  respectively, have been identical so far. 
  Consider as final the states of this additional component in which all sets
  in both partitions are singletons. 
  The purpose of the last component is to be able to only consider runs of
  the synchronized product in which the $k'$ components corresponding to copies
  of~$\Aa$, and the $\ell'$ components corresponding to copies of~$\Bb$, have 
  all distinct runs.
  Similarly, the purpose of the copies of~$\Aa_{k'}$ and~$\Bb_{\ell'}$ is 
  to be able to only consider runs of the synchronized product which record 
  all the $k'$ distinct accepting runs of~$\Aa$ and all of the $\ell'$ 
  distinct accepting runs of~$\Bb$, respectively, on the underlying words.
  Let~$\Cc_{k', \ell'}$ be the resulting finite automaton with $k'+\ell'+3$ 
  components.
  
    The following proposition follows by the construction of 
  automaton~$\Cc_{k', \ell'}$. 
  \begin{proposition}
  \label{proposition:what-C-accepts}
    There are exactly $k'$ distinct runs of $\Aa$ on~$w$ and exactly
    $\ell'$ runs of $\Bb$ on~$w$, if and only if 
    there is an accepting run of~$\Cc_{k', \ell'}$ on~$w$.
  \end{proposition}
  
  Consider the set of spines of~$\Cc_{k', \ell'}$ in which all $k'+\ell'+3$ 
  components of the last state are accepting states; let~$m$ be the size of 
  this set of \emph{accepting spines}.
  For every such accepting spine $\gamma$, we define an instance of vectors
  $\pp^\gamma, \qq^\gamma_1, \ldots, \qq^\gamma_{k'}, \rr^\gamma, \ss^\gamma_1, \ldots, \ss^\gamma_{l'}$.
  If we set $n = |\periods(Q(\gamma))|$ and (arbitrarily) enumerate all simple 
  cycles in~$\periods(Q(\gamma))$ from~$1$ to~$n$, then
  \begin{itemize}
  \item
    $\pp^\gamma$ has $k'$ components: 
    for every $i$, such that $1 \leq i \leq k'$, 
    we set $p^\gamma_i$ to be the product of the probabilities of the transitions 
    in the $i$-th copy of~$\Aa$ in spine~$\gamma$; 
  \item
    $\qq^\gamma_i$ has $n$ components: 
    for every $1 \leq j \leq n$, we set the $j$-th component of $\qq^\gamma_i$
    to be the product of the probabilities of the 
    transitions in $i$-th copy of~$\Aa$ in the $j$-th cycle in the 
    set~$\periods(Q(\gamma))$; 
  \item
    $\rr^\gamma$ has $\ell'$ components: 
    for every $i$, such that $1 \leq i \leq \ell'$, 
    we set $r^\gamma_i$ to be the product of the probabilities of the transitions 
    in the $i$-th copy of~$\Bb$ in spine~$\gamma$; 
  \item
    $\ss^\gamma_i$ has $n$ components: 
    for $1 \leq j \leq n$, we set the $j$-th component of $\ss^\gamma_i$
    to be the product of the probabilities of the 
    transitions, in $i$-th copy of~$\Bb$, in the $j$-th cycle in the 
    set~$\periods(Q(\gamma))$.
  \end{itemize}
In the special case when $k' = 0$ or $l' = 0$ we put $0$ everywhere (which can be understood as a $0$-dimensional vector).

  For an arithmetic expression~$E$ over $n$ variables~$\xx$ 
  indexed by elements of a set~$I$, and for a function 
  $\sigma : I \to \nat$, we write $E[\sigma/\xx]$ for the numerical 
  value of the expression~$E$ in which every occurrence of 
  variable~$x_i$ was replaced by~$\sigma(i)$, for every~$i \in I$. 
  The following proposition follows again by the construction of
  automaton~$\Cc_{k', \ell'}$, taking into account the following 
  observations:
  \begin{itemize}
  \item
    the probability of a run of a probabilistic automaton can be 
    determined from its simple cycle decomposition $(\gamma, \sigma)$, 
    by taking the product of the following:
    \begin{itemize}
    \item 
      the product of the probabilities of the transitions in spine~$\gamma$,
    \item
      for every simple cycle $\omega \in \periods(Q(\gamma))$, the 
      $\sigma(\omega)$-th power of the product of the probabilities of 
      the transitions in~$\omega$;
    \end{itemize}
  \item  
    in an accepting run of~$\Cc_{k', \ell'}$ on a word~$w \in \Sigma^*$, 
    the $k'$ components that correspond to $k'$ copies of~$\Aa$ all follow 
    a distinct run of~$\Aa$ on~$w$, and hence by 
    Proposition~\ref{proposition:what-C-accepts}, $\inter{\Aa}(w)$ is the 
    sum of the probabilities of the $k'$ distinct runs followed by the 
    $k'$ copies of~$\Aa$; 
  \item  
    in an accepting run of~$\Cc_{k', \ell'}$ on a word~$w \in \Sigma^*$, 
    the $\ell'$ components that correspond to $\ell'$ copies of~$\Bb$ all 
    follow a distinct run of~$\Bb$ on~$w$, and hence by 
    Proposition~\ref{proposition:what-C-accepts}, $\inter{\Bb}(w)$ is the 
    sum of the probabilities of the $\ell'$ distinct runs followed by the 
    $\ell'$ copies of~$\Bb$.
  \end{itemize}

  \begin{proposition}
  \label{proposition:S-gives-inter}
    If there is an accepting run~$\rho$ of $\Cc_{k', \ell'}$ on 
    word~$w \in \Sigma^*$, then for every simple cycle decomposition
    $(\gamma, \sigma)$ of~$\rho$, we have
    \begin{itemize}
        \item $\inter{\Aa}(w) = S(\pp^\gamma,
            \qq^\gamma_1,\ldots,\qq^\gamma_{k'})[\sigma/\xx]$ and
        \item $\inter{\Bb}(w) = S(\rr^\gamma, \ss^\gamma_1,\ldots,
            \ss^\gamma_{l'})[\sigma/\xx]$.
    \end{itemize}
  \end{proposition}

  Let $\Delta$ be the set of all tuples $(\pp^\gamma, \qq^\gamma_1, \ldots, \qq^\gamma_{k'}, \rr^\gamma, \ss^\gamma_1, \ldots, \ss^\gamma_{l'})$ given by all accepting spines, in particular $|\Delta|$ is finite.
  We now argue that the two conditions in the statement of Proposition~\ref{proposition:translation}
  are indeed equivalent.
  Let $w \in \Sigma^*$ be a word such that 
  $\inter{\Aa}(w) > \inter{\Bb}(w)$. 
  Let the numbers of distinct accepting runs of~$\Aa$ and~$\Bb$, 
  respectively, on~$w$ be $k'$ and~$\ell'$, respectively. 
  Then, by Proposition~\ref{proposition:what-C-accepts}, there is an 
  accepting run~$\rho$ of~$\Cc_{k', \ell'}$ on~$w$.
  Let $(\gamma, \sigma)$ be a simple cycle decomposition of~$\rho$;
  note that since $\gamma$ is an accepting spine, we have 
  $(\pp^\gamma, \qq^\gamma_1,\ldots,\qq^\gamma_{k'}, \rr^\gamma, \ss^\gamma_{1},\ldots,\ss^\gamma_{l'}) \in \Delta$.  
  It then follows by Proposition~\ref{proposition:S-gives-inter}, 
  that $\sigma$ is a non-negative integer solution of 
  $S(\pp^\gamma, \qq^\gamma_1,\ldots,\qq^\gamma_{k'})[\sigma/\xx] > S(\rr^\gamma, \ss^\gamma_1,\ldots, \ss^\gamma_{l'})[\sigma/\xx]$. 
  
  Conversely, suppose that there is a non-negative integer solution~$\sigma$ of 
  the inequality 
  \[
      S(\pp^\gamma, \qq^\gamma_1,\ldots,\qq^\gamma_{k'})[\sigma/\xx] >
      S(\rr^\gamma, \ss^\gamma_1,\ldots, \ss^\gamma_{l'})[\sigma/\xx],
  \]
  for some quadruple 
  \[
      (\pp^\gamma, \qq^\gamma_1,\ldots,\qq^\gamma_{k'}, \rr^\gamma,
      \ss^\gamma_{1},\ldots,\ss^\gamma_{l'}) \in \Delta.
  \]
  This quadruple is in~$\Delta$ because $\gamma$ is an accepting spine
  of the automaton~$\Cc_{k', \ell'}$ for some $k'$ and~$\ell'$, such 
  that $0 \leq k' \leq k$ and $0 \leq \ell' \leq \ell$.  
  Let $\rho$ be an accepting run of~$\Cc_{k', \ell'}$ that is obtained
  by injecting into~$\gamma$, in some order, $\sigma(\omega)$ copies
  of the simple cycle~$\omega$, for all $\omega \in \periods(Q(\gamma))$;
  let~$w \in \Sigma^*$ be the word underlying the run~$\rho$. 
  By Proposition~\ref{proposition:S-gives-inter}, it follows that 
  $\inter{\Aa}(w) > \inter{\Bb}(w)$. 
\end{proof}

\section{Decidability of containment: the case unambiguous vs. finitely ambiguous}
\label{sec:decidability1vsk}
In this section we will show the more challenging part of
Theorem~\ref{theorem:decidability}, i.e., that the containment problem is
decidable for $\Aa$ unambiguous and $\Bb$ finitely ambiguous.  Our proof is
conditional on the first-order theory of the reals with the exponential function
being decidable. In~\cite{mw96}, the authors show that this is the case if a
conjecture due to Schanuel and regarding transcendental number theory is true.

\begin{theorem}\label{theorem:1-vs-k}
    Determining whether $\inter{\Aa} \leq \inter{\Bb}$ is decidable when
    $\Aa$ is unambiguous and $\Bb$ is finitely ambiguous, assuming Schanuel's
    conjecture is true.
\end{theorem}

\subsection{Integer programming problem with exponentiation}
Given two positive integers $n$ and $\ell$, we define $\mathcal{F}_{n,\ell}$ to
be the set of all the functions $f:\mathbb{R}^n\rightarrow \mathbb{R}$ such that
there exist $\rr \in \mathbb{Q}_{> 0}^\ell$ and $\ss_1, \dots, \ss_\ell \in
\mathbb{Q}_{> 0}^n$ such that $f(\xx) = \sum_{i=1}^{\ell}r_i s_{i,1}^{x_1}
\ldots s_{i,n}^{x_n}$.  Observe that this is just a lifting of the $S(\cdot)$
function, defined in the previous section, to real-valued parameters.  Consider
the following integer programming problem with exponentiation.

\begin{problem}[IP+EXP]\hfill
    \begin{itemize}
    \item \textbf{Input:} Three positive integers $n$, $\ell$ and $m$, a
        function 
    $f\in \mathcal{F}_{n,\ell}$, a matrix $M\in  \mathbb{Z}^{m\times n}$, and a
        vector $\cc \in \mathbb{Z}^m$.
    \item \textbf{Question:} Does there exist $\xx \in \mathbb{Z}^n$
        such that $f(\xx) < 1$ and $M\xx < \cc$?
    \end{itemize}
\end{problem}

In the sequel, we will show that the above problem is decidable.
\begin{theorem} \label{thm:main}
    The IP+EXP problem is decidable, assuming Schanuel's conjecture is true.
\end{theorem}
Theorem~\ref{theorem:1-vs-k} is a direct corollary of Theorem~\ref{thm:main}.

\begin{proof}[Proof of Theorem~\ref{theorem:1-vs-k}]
    Proposition~\ref{proposition:translation} shows that, in order to prove
    Theorem~\ref{theorem:1-vs-k}, it is sufficient to decide, given an integer
    $n$ and positive rational numbers $p, r_i, q_{j}, s_{i,j}$ for
    $i\in\{1,\ldots,\ell\}$, $j\in\{1,\ldots,n\}$, whether there exist
    $x_1,\ldots,x_n \in \nat$ such that $p q_{1}^{x_1} \cdots q_{n}^{x_n} >
    \sum_{i=1}^\ell r_i s_{i,1}^{x_1} \cdots s_{i,n}^{x_n}$ or equivalently,
    whether there exist $x_1,\ldots,x_n \in \nat$ such that:
    \begin{equation} \label{eq:finite-unamb}
        \sum_{i=1}^\ell r_ip^{-1} (s_{i,1}q_1^{-1})^{x_1} \cdots
        (s_{i,n}q_n^{-1})^{x_n} < 1.
    \end{equation}

    Define $f:\mathbb{R}^n\rightarrow \mathbb{R}$ such that $f(\xx) =
    \sum_{i=1}^{\ell}r_ip^{-1} (s_{i,1}q_1^{-1})^{x_1} \cdots
    (s_{i,n}q_n^{-1})^{x_n}$.  Then, inequality~\eqref{eq:finite-unamb} becomes
    $f(\xx)<1$.  We can now apply Theorem~\ref{thm:main} with $m$ set to be $n$;
    $M$, to be $-Id$, where $Id$ is the identity matrix; and $\cc$ to be the
    null vector.
\end{proof}

Since the IP+EXP problem is semi-decidable (indeed, we can enumerate the vectors
$\xx$ in $\mathbb{Z}^n$ to find one satisfying the conditions), it will suffice
to give a semi-decision procedure to determine whether the inequalities
$f(\xx)<1 \wedge M\xx<\cc$ have no integer solution.  We give now such a
procedure.

\subsection{Semi-decision procedure for the complement of IP+EXP}
Consider as input for the IP+EXP problem three positive integers $n$, 
$\ell$, $m$, a function 
$f\in \mathcal{F}_{n,\ell}$, 
a matrix $M\in  \mathbb{Z}^{m\times n}$, and a vector $\cc \in \mathbb{Z}^m$.
Denote by $X$ the set of \textit{real} solutions of the problem, 
i.e., the set of vectors
\(
    X = \{ \xx \in \mathbb{R}^n \st f(\xx)<1 \wedge M\xx < \cc\}. 
\)

\paragraph{Proc($n$, $\ell$, $m$, $f$, $M$, $\cc$)}
\begin{enumerate}
    \item Search for a non-zero vector $\dd\in \mathbb{Z}^n$ and
        $a,b\in\mathbb{Z}$ such that $\{\dd^\top \xx \st \xx \in X\} \subseteq
        [a,b]$. Set $i = a$.
    \item If $i>b$, then stop and return $\texttt{YES}$. Otherwise, let $Y_i$ be
        the set of vectors $\xx \in  \mathbb{Z}^n$ satisfying $d_1x_1 + \cdots +
        d_nx_n = i$.  If $Y_i$ is empty, then increment $i$ and start again from
        step $2$.  Otherwise:
        \begin{enumerate}
            \item  Compute $N \in \int^{n\times (n-1)}$ and $\hh\in
                \mathbb{Z}^n$ such that $Y_i = \{ N\boldsymbol{y} + \hh \st \yy
                \in \mathbb{Z}^{n-1}\}$.
            \item If $n-1=0$ and $f(\hh)<1 \wedge M\hh <\cc$ then return
                \texttt{NO}, otherwise increment $i$ and start again from step
                $2$.
            \item If $n-1>0$ then recursively call \textbf{Proc($n-1$, $\ell$,
                $m$, $f'$, $M'$, $\cc'$)}, where $f' \in \mathcal{F}_{n-1,\ell}$
                is defined as $f'(\yy) = f(N\yy + \hh)$; $M'\in
                \mathbb{Z}^{m\times (n-1)}$, as $M' = MN$; and $\cc' \in
                \mathbb{Z}^m$, as $\cc - M\hh$. If the procedure stops and
                returns \texttt{YES} then increment $i$ and start again from
                step $2$. If the procedure stops and returns \texttt{NO} then
                return \texttt{NO}.
    \end{enumerate}
\end{enumerate}

\begin{lemma} \label{lemma:procedure}
    The above semi-decision procedure stops and
    outputs \texttt{YES} if and only if there is no integer valuation of 
    $\xx$ that satisfies the constraints, i.e. $X\cap \int^n$ is empty.
\end{lemma}

We prove this lemma in Section~\ref{sec:lemma:procedure}. Before,
let us shortly comment on both steps of the procedure.

\paragraph*{Step 1 of the procedure}
First, notice that the only step which might not terminate in a call to our
procedure is step~$1$.  Indeed, once $\dd$, $a$, and $b$ are fixed, there are
only finitely many integers $i \in [a,b]$ that have to be considered in step
$2$.

Moreover, for each integer vector $\dd\in \mathbb{Z}^n$ and $a,b\in\mathbb{Z}$,
the inclusion $\{\dd^\top \xx \st \xx \in X\} \subseteq [a,b]$ that needs to be
checked in step $1$ can be formulated as a decision problem in the first-order
logic over the structure $(\mathbb{R},+,\times,\exp)$.  Since this structure has
a decidable first-order theory subject to Schanuel's conjecture~\cite{mw96}, the
inclusion can be decided for each fixed $\dd$, $a$, and $b$.


\paragraph*{Step 2 of the procedure}
For fixed $\dd$, $a$, and $b$, one can compute in a standard way
the set of all integer solutions $Y_i$ 
(see, e.g.,~\cite{cohen93}), as we now explain. By
performing elementary column operations,
find a $n \times n$ unimodular (i.e. with determinant equal to $1$ or $-1$) 
integer matrix $U$ such that
\[
    (d_1,\dots,d_n) U =  (g, 0, \dots, 0),
\]
where $g=\gcd(d_1,\dots,d_n)$. Recall that $Y_i$ is the set
of integer solutions of
$\dd^\top\xx = i$.  We apply the change of variables $U\yy=\xx$ to it,
where $\yy=(y_1,\dots,y_n)$, to obtain $\dd^\top U \yy = i$.  Since $\dd^\top  U
\yy = (g,0,\dots,0)\yy = gy_1$, the transformed equation is $gy_1 = i$ and the
matrix $U$ gives a one-to-one correspondence between integer solutions $\yy$ of
the transformed equation and solutions $\xx \in Y_i$.  Now the
transformed equation has a solution if and only if $g$ divides $i$,  in which
case $y_1 = i/g$. Furthermore, in this
case a general solution from $Y_i$ can be written in the form
$\xx=N\yy'+\hh$ for $N$ a $n \times (n-1)$ integer matrix and $\hh\in
\mathbb{Z}^n$ (both derived from $U$) and $\yy'=(y_2,\dots,y_n)$.

\subsection{Proof of Lemma~\ref{lemma:procedure}}\label{sec:lemma:procedure}
The proof of Lemma~\ref{lemma:procedure} relies on the two following lemmas. The
first one is the most technical contribution of the paper and is proved in
Section~\ref{section:proof-main-lemma}. It ensures termination of step~$1$ in
the procedure when there is no integer solution. 

\begin{lemma} \label{lemma:proc1}
    If the set $X$ contains no integer point then there must exist a non-zero
    integer vector $\dd\in \mathbb{Z}^n$ and $a,b\in\mathbb{Z}$ such that
    $\{\dd^\top \xx \st \xx \in X\} \subseteq [a,b]$. 
\end{lemma}

This second lemma guarantees
that the recursive calls in step~$2$ guarantee the correct output.
\begin{lemma}
\label{lemma:proc2} 
    Given a non-zero vector $\dd\in \mathbb{Z}^n$ and an integer $i$, there
    exists $\xx \in \mathbb{Z}^n$ such that \(f(\xx)<1 \land M\xx < \cc \land
    \dd^\top \xx = i \) if and only if there exists $\yy \in \mathbb{Z}^{n-1}$
    such that \(f'(\yy)<1 \land M'\yy < \cc' \) where $f'$, $M'$ and $\cc'$ are
    as defined in the procedure.
\end{lemma}
\begin{proof}
    We want to prove that given a non-zero vector $\dd\in \mathbb{Z}^n$ and an
    integer $i$, there exists $\xx \in \mathbb{Z}^n$ such that \(f(\xx)<1 \land
    M\xx < \cc \land \dd^\top \xx = i \) if and only if there exists $\yy \in
    \mathbb{Z}^{n-1}$ such that \(f'(\yy)<1 \land M'\yy < \cc' \) where $f'$,
    $M'$ and $\cc'$ are as defined in the procedure.  Recall that $Y_i$ is the
    set of vectors $\xx$ such that $\dd^\top \xx = i$ and that $Y_i = \{
    N\boldsymbol{y} + \hh \st \yy \in \mathbb{Z}^{n-1}\}$ for some $N \in
    \int^{n\times (n-1)}$ and $\hh\in \mathbb{Z}^n$. 

    Let $\xx \in \mathbb{Z}^n$ such that 
    \(
        f(\xx)<1 \land M\xx < \cc \land \dd^\top \xx = i. 
    \)

    Then $\xx \in Y_i$ and thus there is $\yy\in \mathbb{Z}^{n-1}$ such that
    $\xx = N\yy + \hh$.  We have: $f'(\yy) = f(N\yy + \hh) = f(\xx) <1$ and
    $M'\yy = MN\yy = M(\xx - \hh) = M\xx - M\hh < \cc -M\hh = \cc'$.

    Conversely, consider $\yy\in \mathbb{Z}^{n-1} \cap Y_i$ such that
    \(f'(\yy)<1 \land M'\yy < \cc' \). Let $\xx = N\yy + \hh$. Then $\xx \in
    Y_i$ and thus $\dd^\top \xx = i$. Moreover, $f(\xx) = f(N\yy + \hh) =
    f'(\yy) < 1$ and $ M\xx = M(N\yy + \hh) = MN\yy + M\hh = M'\yy + M\hh < \cc'
    + M\hh = \cc$.
\end{proof}

We prove Lemma~\ref{lemma:procedure}.

\paragraph{First direction: when the procedure returns \texttt{YES}} Suppose
first that the semi-decision procedure stops and outputs \texttt{YES}. Then
there exist a non-zero vector $\dd\in \mathbb{Z}^n$ and $a,b\in\mathbb{Z}$ such
that $\{\dd^\top \xx \st \xx \in X\} \subseteq [a,b]$ as in step $1$, and for
all integers $i \in [a,b]$, one of the following situations occurs:
\begin{enumerate}
    \item $Y_i$ is empty,
    \item $n-1 = 0$, $Y_i = \{\hh\}$ as defined in step $2.a$ but $\hh$ is not
        an integer solution of the problem,
    \item $n-1 > 0$ and the recursive call stops and outputs \texttt{YES}.
\end{enumerate}
By definition of $\dd$, in order to prove that there is no integer solution of
the problem, we need to show that in all those cases, and for all $i \in [a,b]$,
$Y_i \cap X = \emptyset$.  It is clear for items $1$ and $2$ and we use
Lemma~\ref{lemma:proc2} and an induction for item $3$.

\paragraph{Second direction: when $X \cap \mathbb{Z}^n = \emptyset$}
If there is no integer solution then by Lemma~\ref{lemma:proc1}, there must
exist a non-zero vector $\dd\in \mathbb{Z}^n$ and $a,b\in\mathbb{Z}$ such that
$\{\dd^\top \xx \st \xx \in X\} \subseteq [a,b]$ as in step $1$.  Moreover, for
any of those choices, if for an integer $i \in [a,b]$, the set $Y_i$ of vectors
$\xx \in  \mathbb{Z}^n$ satisfying $d_1x_1 + \cdots + d_nx_n = i$ is non-empty,
then, 
\begin{enumerate}
    \item if $n = 1$, then $\hh$ as defined in step $2.a$, is not a solution of
        the problem (by hypothesis) and thus the procedure stops and returns
        \texttt{YES},
    \item if $n > 1$, we use Lemma~\ref{lemma:proc2} and, by induction, the
        recursive call must return \texttt{YES}.
\end{enumerate}

\subsection{Proof of Lemma~\ref{lemma:proc1}} \label{section:proof-main-lemma}
Fix three positive integers $n$, $\ell$, $m$, a function $f\in
\mathcal{F}_{n,\ell}$, a matrix $M\in  \mathbb{Z}^{m\times n}$, and a vector
$\cc \in \mathbb{Z}^m$.  Recall that we denote by $X$ the set of vectors
\begin{equation*}
    X = \{ \xx \in \mathbb{R}^n \st f(\xx)<1 \wedge M\xx < \cc\}. 
\end{equation*} 
We want to prove that if the set $X$ contains no integer point then there must
exist a non-zero integer vector $\dd\in \mathbb{Z}^n$ and $a,b\in\mathbb{Z}$
such that $\{\dd^\top \xx \st \xx \in X\} \subseteq [a,b]$. 

We will use the following corollary of Kronecker's theorem on
simultaneous Diophantine approximation.  It generalises the fact that
any line in the plane with irrational slope passes arbitrarily close
to integer points in the plane.
\begin{proposition}\cite[Corollary 2.8]{kp97}.
    Let $\uu,\uu_1,\dots,\uu_s$ be vectors in $\mathbb{R}^n$.  Suppose
    that for all $\dd\in\mathbb{Z}^n$ we have $\dd^\top\uu=0$ whenever
    $\dd^\top \uu_1 = \dots = \dd^\top \uu_s = 0$.  Then for all
    $\varepsilon>0$ there exist real numbers $\lambda_1,\dots,\lambda_s
    \geq 0$ and a vector $\vv \in \mathbb{Z}^n$ such that
    \(
        \left\|\uu + \sum_{i=1}^s \lambda_i \uu_i - \vv \right\|_\infty \leq
        \varepsilon.
    \)
\label{pro:kronecker}
\end{proposition}

By definition, there exist vectors $\rr \in \mathbb{Q}_{> 0}^\ell$ and $\ss_1,
\dots, \ss_\ell \in \mathbb{Q}_{> 0}^n$ such that $f(\xx) = \sum_{i=1}^{\ell}r_i
s_{i,1}^{x_1} \ldots s_{i,n}^{x_n}$.  Let $\aa \in \real^\ell$ and $\bb_i \in
\real^n$ be defined by $a_i = \log(r_i)$ and $\bb_i = (\log(s_{i,1}), \ldots,
\log(s_{i,n}))$. We can then rewrite $f(\xx)$ as follows
\[
    f(\xx) = \exp(\bb_1^\top \xx + a_1) + \dots + \exp(\bb_\ell^\top \xx +
    a_\ell).
\]

Let us now consider the cone
\begin{equation}\label{eq:rec-cone}
    C = \left\{ \xx \in \mathbb{R}^n \:\middle|\: \bb_1^\top\xx \leq 0
    \wedge \dots \wedge \bb_\ell^\top\xx \leq 0 \wedge M\xx\leq 0 \right\}. 
\end{equation}
It is easy to see that $X + C \subseteq X$.

\begin{lemma}
    Suppose that $X$ is non-empty and that no non-zero integer vector in
    $\mathbb{Z}^n$ is orthogonal to $C$.  Then  
    $X \cap \mathbb{Z}^n$ is non-empty. 
\label{lemma:alternative}
\end{lemma}
\begin{proof}
    Let $\uu \in X$.  Since $X$ is open, there exists $\varepsilon>0$ such that
    the open ball $B_\varepsilon(\uu)$ is contained in $X$.  We therefore have
    that $B_\varepsilon(\uu) + C \subseteq X$.

    We will apply Proposition~\ref{pro:kronecker} to show that
    $B_\varepsilon(\uu) + C$ contains an integer point and hence that $X$
    contains an integer point.  To this end, let vectors $\uu_1,\dots,\uu_s \in
    C$ be such that $\vspan\{\uu_1,\dots,\uu_s\}=\vspan(C)$.  Then
    no non-zero vector in $\mathbb{Z}^n$ is orthogonal to $\uu_1,\dots,\uu_s$.
    By Proposition~\ref{pro:kronecker}, there exist real numbers
    $\lambda_1,\dots,\lambda_s \geq 0$ and an integer vector $\vv \in
    \mathbb{Z}^n$ such that
    \(
        \left\|\uu + \sum_{i=1}^s \lambda_i \uu_i - \vv \right\|_\infty \leq
        \varepsilon.
    \)
    Thus, $\vv \in B_{\varepsilon}(\uu) +C \subseteq X$.
\end{proof}

The contrapositive of the above result states that if $X$ contains no integer
point, then there must exist an integer vector that is orthogonal to $C$. For
the desired result, it remains for us to prove the boundedness claim.

\begin{lemma}
    Suppose that $\dd \in \mathbb{Z}^n$ is orthogonal to the cone $C$.
    Then $\{ \dd^\top \uu \st \uu \in  X\}$ is bounded.
\label{lemma:orthog}
\end{lemma}
\begin{proof}
    Define the ``enveloping polygon'' of $X$ to be 
    \[
        \widehat{X} =
        \left\{
            \xx \in \mathbb{R}^n \:\middle|\:
            \bb_1^\top \xx + \aa_1 \leq 0 \wedge \dots \wedge
            \bb_\ell^\top \xx + \aa_\ell \leq 0 \wedge  M\xx \leq \cc
        \right\}. 
    \]
    Clearly it holds that $X\subseteq \widehat{X}$.  Moreover, by the
    Minkowski-Weyl decomposition theorem we can write $\widehat{X}$ as a sum
    $\widehat{X}=B+C$ for $B$ a bounded polygon and $C$ the cone defined in
    (\ref{eq:rec-cone}).  Since $\dd$ is orthogonal to $C$ by assumption, it
    follows that $\{ \dd^\top \uu \st \uu \in \widehat{X}\}=\{\dd^\top \uu \st
    \uu \in B\}$ is bounded and hence $\{ \dd^\top \uu \st \uu \in X\}$ is
    bounded.  The result immediately follows.
\end{proof}

We can now complete the proof of Lemma~\ref{lemma:proc1}.
\begin{proof}[Proof of Lemma~\ref{lemma:proc1}]
By Lemma~\ref{lemma:alternative} there exists a non-zero integer
vector $\dd\in\mathbb{Z}^n$ such that $\dd$ is orthogonal to the cone
$C$ defined in (\ref{eq:rec-cone}).  Then by
Lemma~\ref{lemma:orthog} we obtain that $\{ \dd^\top \uu \st \uu \in
X\}$ is contained in a bounded interval.
\end{proof}

\section{Undecidability for linearly ambiguous automata}
\label{sec:undecidability}
In this section we prove Theorem~\ref{theorem:undecidability}. That is, we argue
that the emptiness and containment problems are undecidable for the class of linearly ambiguous
PA. In fact, We will prove a more general result in
Proposition~\ref{proposition:automaton}.

The proof is done by a reduction from the halting problem for two-counter
machines. The reduction resembles the one used to prove undecidability of the
comparison problem for another quantitative extension of Boolean automata:
max-plus automata~\cite{Colcombet, abk11}.

\newcommand{\A}{\mathcal{A}}
\newcommand{\B}{\mathcal{B}}
\newcommand{\Ap}{\mathcal{A}'}
\newcommand{\Bp}{\mathcal{B}'}
\newcommand{\funcAp}{[\![\mathcal{A}']\!]}
\newcommand{\funcBp}{[\![\mathcal{B}']\!]}
\newcommand{\Azero}{\A_0}
\newcommand{\Bzero}{\B_0}
\newcommand{\funcAzero}{[\![\Azero]\!]}
\newcommand{\funcBzero}{[\![\Bzero]\!]}
\newcommand{\Aone}{\A_1}
\newcommand{\Bone}{\B_1}
\newcommand{\funcAone}{[\![\Aone]\!]}
\newcommand{\funcBone}{[\![\Bone]\!]}
\renewcommand{\Aa}{\A_2}
\newcommand{\Ba}{\B_2}
\newcommand{\funcAa}{[\![\Aa]\!]}
\newcommand{\funcBa}{[\![\Ba]\!]}
\newcommand{\Ab}{\A_3}
\renewcommand{\Bb}{\B_3}
\newcommand{\funcAb}{[\![\Ab]\!]}
\newcommand{\funcBb}{[\![\Bb]\!]}
\newcommand{\Az}{\A_6}
\newcommand{\Bz}{\B_6}
\newcommand{\funcAz}{[\![\Az]\!]}
\newcommand{\funcBz}{[\![\Bz]\!]}

\renewcommand{\funcA}{\inter{\mathcal{A}}}
\renewcommand{\funcB}{\inter{\mathcal{B}}}

\newcommand{\alphabet}{\Sigma}

\subsection{Two-counter machines}
\emph{Two-counter machines (or Minsky machines)} can be defined in several ways,
all equivalent in terms of expressiveness. We use here the following
description: A two-counter machine is a deterministic finite-state machine with
two counters that can be incremented, decremented, or tested for $0$. Formally, 
it is given by a tuple $(Q, T^+_1, T^+_2, T^-_1, T^-_2, q_{init},
q_{halt})$ where:
\begin{itemize} 
    \item $Q$ is a finite set of states.
    \item $T^+_1$ (resp.\ $T^+_2$) is a subset of $Q^2$. If $(p,q) \in T^+_1$
        (resp.\ $T^+_2$) then there is a transition from the state $p$ to the
        state $q$ which increments the first counter (resp.\ second counter).
    \item $T^-_1$ (resp.\ $T^-_2$) is a subset of $Q^3$. If $(p,q,r) \in T^-_1$
        (resp.\ $T^-_2$) then there is a transition from the state $p$ which goes
        to the state $q$ if the current value of the first (resp.\ second)
        counter is $0$ (it does not change the counters), and which goes to the
        state $r$ otherwise and decrements the first (resp.\ second) counter.
    \item $q_{init} \in Q$ is the initial state and $q_{halt} \in Q$ is the
        final state such that there is no outgoing transition from $q_{halt}$
        (for all transitions $(q,p) \in T^+_1 \cup T^+_2$ or $(q,p,r) \in T^-_1
        \cup T^-_2$, $q \neq q_{halt}$). We also assume that $q_{init} \neq
        q_{halt}$.
\end{itemize}
Moreover the machine is deterministic:  for every state there  is  at  most  one
action that  can  be  performed, i.e. for all $q\in Q$, there is at most one
transition of the form $(q,p)$ or $(q,p,r)$ in $T^+_1 \cup T^+_2 \cup T^-_1 \cup
T^-_2$ and $T^+_1 \cap T^+_2 = \emptyset$ and $T^-_1 \cap T^-_2 = \emptyset$.

The  semantics  of  a  two-counter  machine  are  given  by  means  of  the
valuations  of  the counters that are pairs of non-negative integers.  An
execution with counters initialised to $(n^0_1,n^0_2)$ is a sequence of
compatible transitions and valuations denoted by
$$
(n^0_1,n^0_2)
\xrightarrow{t_1} (n^1_1,n^1_2) \xrightarrow{t_2} (n^2_1,n^2_2)
\xrightarrow{t_3} \dotsm \xrightarrow{t_k} (n^k_1,n^k_2)
$$
such that:
\begin{itemize}
\item for all $i\in \{1,\ldots,k\}$, if $t_i \in T^+_1$ (resp.\ $T^+_2$), then
    $n_1^i = n_1^{i-1}+1$ and $n_2^i = n_2^{i-1}$ (resp.\ $n_1^i = n_1^{i-1}$ and
    $n_2^i = n_2^{i-1}+1$);
\item for all $i\in \{1,\ldots,k\}$, if $t_i \in T^-_1$ (resp.\ $T^-_2$), then
    $n_1^i = n_1^{i-1}=0$ or $n_1^i = n_1^{i-1}-1$ and $n_2^i = n_2^{i-1}$
    (resp.\ $n_2^i = n_2^{i-1}=0$ or $n_2^i = n_2^{i-1}-1$ and $n_1^i =
    n_1^{i-1}$);
\item for all $i\in \{1,\ldots,k-1\}$, if $t_i = (p_i,q_i) \in T^+_1 \cup T^+_2$
    then $t_{i+1} \in \{q_i\} \times (Q\cup Q^2)$;
\item for all $i\in \{1,\ldots,k-1\}$, if $t_i = (p_i,q_i,r_i) \in T^-_1 \cup
    T^-_2$ and $n_{i-1} = 0$ then $t_{i+1} \in \{q_i\} \times (Q\cup Q^2)$,
    otherwise if $n_{i-1} \neq 0$ then $t_{i+1} \in \{r_i\} \times (Q\cup
    Q^2)$.
\end{itemize}
We say that the machine \emph{halts} if there is a (unique) execution with counters
initialised to $(0,0)$ starting in $q_{init}$ reaching the state $q_{halt}$.

\begin{proposition}[\cite{minsky67}]
The halting problem for two-counter machines is undecidable. 
\end{proposition}

\subsection{Reduction from the halting problem for two-counter machines}

\begin{proposition}\label{proposition:automaton}
    Given a two counter machine, one can construct a linearly ambiguous
    probabilistic automaton $\A$ such that the machine halts if and only if
    there exists a word $w$ such that $\funcA(w) \geq \frac{1}{2}$
    (resp.\ $>$, $\leq$, $<$).
\end{proposition}

Theorem~\ref{theorem:undecidability} is an immediate corollary of
Proposition~\ref{proposition:automaton} since it is trivial to construct a PA
that outputs probability $\frac{1}{2}$ for all words.

We follow these steps:
\begin{enumerate}
\item We construct two linearly ambiguous PA $\A$ and $\B$
    such that the machine halts if and only if there is a word $w$
    such that $\funcA(w) \leq \funcB(w)$. 
\item From $\A$ and $\B$, we construct $\Ap$ and $\Bp$, also linearly ambiguous,
    such that the machine halts if and only if there is a word  $w$
    such that $\funcAp(w) < \funcBp(w)$. 
\item We show that the functions $1-\funcA$, $1-\funcB$, $1-\funcAp$,
    $1-\funcBp$ are also computed by linearly ambiguous PA.
\end{enumerate}

We show that Proposition~\ref{proposition:automaton} follows from steps 1, 2, and 3. We
have that for all words $w$
\begin{align*}
	\funcA(w) \leq \funcB(w) \quad 
      & \iff \quad \frac{1}{2}\funcA(w) + \frac{1}{2}(1-\funcB(w)) \leq \frac{1}{2} \\
      & \iff \quad \frac{1}{2}\funcB(w) + \frac{1}{2}(1-\funcA(w)) \geq \frac{1}{2}. 
\end{align*}
Since
$1-\funcA$ and $1-\funcB$ are linearly ambiguous by 3, then, $\frac{1}{2}\funcA
+ \frac{1}{2}(1-\funcB)$ and $\frac{1}{2}\funcB + \frac{1}{2}(1-\funcA)$
are also linearly ambiguous. From
1, we get that the two variants with non-strict inequalities of the problem are
undecidable. We proceed similarly using $\Ap$ and $\Bp$ to prove
undecidability of the two other variants with strict inequalities. 

\subsubsection*{Step 1}
Let $T = T^+_1 \cup T^+_2 \cup T^-_1 \cup T^-_2$ and $\Sigma = \{a,b\} \cup T$.
The idea is to encode the executions of the two-counter machine into words over
the alphabet $\Sigma$.
A block $a^m$ (resp.\ $b^m$) encodes the fact that the value of 
the first (resp.\ second) counter is~$m$. 
For example, given $t \in T_1^+$ and $t' \in T_2^{-}$, a word 
$a^nb^m t a^{n+1}b^m t' a^{n+1}b^{m'}$, encodes an execution 
starting with value $n$ in the first counter and $m$ in the 
second counter. 
Transition $t$ then increases the value of the first counter 
to $n+1$ without changing the value of the second one. The configuration is thus 
encoded by the infix $a^{n+1}b^m$. 
Next, transition $t'$ is taken, and either $m'=m=0$ or $m' = m-1$. 
Moreover, if $t = (p,q)$ and $t' = (r,s,u)$ then $q = r$ 
(i.e., the states between transitions have to match).

The PA $\A$ and $\B$ are constructed in such a way that for all words $w$ it
holds that
\[
	\begin{cases}
    \funcA(w) = \funcB(w) & \text{if } w
    	\text{ represents a valid halting execution of the machine}\\
    \funcA(w) > \funcB(w) & \text{otherwise.}
    \end{cases}
\]
The automata $\A$ and $\B$ are constructed as a weighted sum of seven PA,
each of them checking some criteria that a word $w$ should fulfill 
in order to represent a valid halting execution.

\paragraph*{Automaton $\Azero$}
Conditions such as asking that the encoded execution start in $q_{init}$, end in
$q_{halt}$, represent a valid path in the machine (with respect to the
states), and that the encoding be of the good shape, i.e.\ contain alternating blocks of $a$'s, blocks of
$b$'s and letters from $T$, are all regular conditions that can thus be checked
by a (deterministic Boolean) automaton. The exhaustive list of such conditions
and their explanations are given below. We then define $\Azero$ to be a
deterministic PA such that $\funcAzero(w) = 0$ if and only
if $w$ satisfies all these regular conditions and $\funcAzero(w) = 1$ otherwise.

We give here a precise description of the conditions checked by $\Azero$.
\begin{enumerate}
\item The word $w$ belongs to $T_{init}((a^*b^*)T)^*$ where $T_{init}$ is the subset of $T$ of the transitions started in $q_{init}$.
\item The word $w$ represents an execution ending in $q_{halt}$, i.e.\ it either
    ends
    \begin{itemize}
        \item with a letter $(q, q_{halt})$,
        \item with a word of the form $t a^n (q, q_{halt},r)$ where $t\in
            T$, $(q, q_{halt},r) \in T^-_2$ and $n$ is a non-negative
            integer (resp. $t b^n (q, q_{halt},r)$ where $t\in T$, $(q,
            q_{halt},r) \in T^-_1$ and $n$ a non-negative integer),
        \item or with a word of the form $t a^n b^m (q, r, q_{halt})$ where
            $t\in T$, $(q, r, q_{halt}) \in T^-_2$ and $m$ a positive
            integer (resp. $t a^n b^m (q, r, q_{halt})$ where $t\in T$,
            $(q, r, q_{halt}) \in T^-_1$ and $n$ a positive integer).
    \end{itemize}
\item The transitions are state-compatible, i.e.\ if $w$ contains a factor $(p,q)a^nb^m t$ 
	with $t \in T$ then $t$ starts in $q$ and if $w$ contains a factor $t' a^nb^m 
    (p,q,r) a^{n'}b^{m'} t$ with $t,t' \in T$ and $(p,q,r) \in T^-_1$ (resp.\ $T^-_2$)
    then $t$ starts in $q$ if $n=0$ and in $r$ if $n>0$ (resp.\ $t$ starts in $q$ if $m=0$
    and in $r$ if $m>0$).
\item We also check that if in the execution represented by the word, at some
    point the value in the first (resp.\ second) counter is $0$ and a transition
    from $T^-_1$ (resp.\ $T^-_2$) is taken then the value in the counter is still
    $0$ after the transition. In terms of words, this means that if $tb^mt'a^nb^{m'}t''$ is a factor of the word with $t,t'' \in T$ and $t' \in T^-_1$ then $n=0$ (and similarly for the second counter). 
\end{enumerate}  

The automaton $\Azero$ will make sure that $\inter{\mathcal{A}}(w) = 0$ only if 
$w$ is \emph{proper}, i.e.\ of the good shape as given above. We are now left to check that
the counters are properly incremented and decremented.

\paragraph*{Automata $\Aone$ and $\Bone$}
The automata $\Aone$ and
$\Bone$ check that a proper word encodes an execution where the first counter
is always correctly incremented after reading transitions from $T_1^+$.
Consider the automaton $\Cc(x,y,z)$ in Figure~\ref{figure:undec-autA}. It is
parameterised by three probability variables $x,y,z > 0$. The parameter $x$ is
the probability used by the initial distribution, and parameters $y$ and $z$ are
used by some transitions.

\begin{figure}[!htbp]
\begin{center}
    \begin{tikzpicture}[scale=0.5]
    \tikzset{
        initial text={\small{$x$}},
    }

    \node[state, initial] (q_1) at (-8,0) {};
    \node[state, initial] (q_2) at (-3,0) {};
    \node[state,accepting] (q_3) at (3,0) {};
    \node[state, accepting] (q_4) at (8,0) {};

    \path[trans]
    (q_1) edge [loop above] node         {\small{$T:\frac{1}{2}$}} ();
    \path[trans]     (q_1) edge [loop below] node         {\small{$a,b:1$}} ();
    \path[trans]
    (q_1) edge [bend left] node [above] {\small{$T:\frac{1}{2}$}}(q_2);
    \path[trans]     (q_2) edge [loop above] node         {\small{$a:y$}} ();
    \path[trans]     (q_2) edge [loop below] node         {\small{$b:1$}} ();
    \path[trans]     (q_2) edge  node [above] {\small{$T_1^{+}:y$}}(q_3);
    \path[trans]     (q_3) edge [loop above] node         {\small{$a:z$}} ();
    \path[trans]     (q_3) edge [bend left] node [above] {\small{$T,b:1$}}(q_4);
    \path[trans]
    (q_4) edge [loop above] node         {\small{$\Sigma:1$}} ();
    \end{tikzpicture}
\end{center}
\caption{\label{figure:undec-autA} Gadget automaton $\Cc(x,y,z)$ used to check if
the first counter is incremented properly.}
\end{figure}

We only take into consideration proper words as given by the automaton $\Azero$.
Notice that the only
non-deterministic transitions in $\Cc(x,y,z)$ are the ones going out from the
the leftmost state upon reading letters from~$T$.  It follows that $\Cc(x,y,z)$
is linearly ambiguous.  In fact, for every position in $w$ labelled by an
element $t$ from $T^+_1$ there is a unique accepting run that first reaches a
final state upon reading $t$.  By construction, we have
\begin{align} \label{eq:valueA}
    \inter{\Cc(x,y,z)}(w) & \quad = \quad \sum_{t_i \in T^+_1}
    x\left(\frac{1}{2}\right)^{i-1}y^{n_{i}+1}z^{n_{i+1}}.
\end{align}

Let $x = \frac{1}{2}$, $y = 1$ and $z = \frac{1}{4}$.  We define $\Bone$ as
$\Cc(x,x,x)$ and $\Aone$ as a weighted sum of $\Cc(x,y,z)$ and $\Cc(x,z,y)$ with
weights $(\frac{1}{2},\frac{1}{2})$. Since $\Cc(\cdot,\cdot,\cdot)$ is linearly
ambiguous the obtained automata are also linearly ambiguous.  We prove that
$\inter{\Aone}(w) = \inter{\Bone}(w)$ only if $n_i + 1 = n_{i+1}$ for all $i$
such that $t_i \in T_1^{+}$ and $\inter{\Aone}(w) > \inter{\Bone}(w)$ otherwise.

By~\eqref{eq:valueA} it suffices to show that for every $i$ it holds that
\[
    \left(\frac{1}{2}\right)^{n_{i}+1 + n_{i+1}} \quad \le \quad 
    \frac{1}{2}\left(\left(\frac{1}{4}\right)^{n_{i+1}}
    + \left(\frac{1}{4}\right)^{n_{i}+1} \right)
\]
and that the equality holds only if $n_i + 1 = n_{i+1}$.
Let $p = \frac{1}{2}^{n_{i}+1}$ and $q = \frac{1}{2}^{n_{i+1}}$ then this reduces to
$$
pq \le \frac{1}{2}\left( p^2 + q^2 \right).
$$
This is true for every $p,q$ and moreover the equality holds if and only if $p = q$, which is equivalent to $n_i + 1 = n_{i+1}$.
We conclude with the following remark that will be useful for Step 2.

\begin{remark}
If $\inter{\Aone}(w) > \inter{\Bone}(w)$ then $\funcAone(w) \geq \funcBone(w) + \left(\frac{1}{2}\right)^{2(|w|+1)}$.
To obtain this bound notice that the probabilities and initial distribution of the automata $\Aone$ and $\Bone$ are $0$, $1$, $\frac{1}{2}$ and $\frac{1}{4}$. Hence for every word $w$ the probability assigned to it is either $0$ or a multiple of $\left(\frac{1}{4}\right)^{|w|+1}$. It follows that if $\Aone$ and $\Bone$ assign different probabilities to $w$ then the difference is at least $\left(\frac{1}{4}\right)^{|w|+1}$, which proves the remark.
\end{remark}

\newcommand{\Dd}{\mathcal{D}}

\paragraph*{Automata $\Aa$, $\Ba$, \ldots}

Similarly, we construct automata $\Aa$, $\Ba$, $\Ab$, $\Bb$, \ldots $\Az$, $\Bz$ to check the other criteria: 
\begin{itemize}
    \item the value of the second counter is correctly incremented when taking a
        transition from $T^+_2$,
    \item the value of the first (resp.\ second) counter remains the same when
        using a transition from $T^+_2 \cup T^-_2$ (resp.\ $T^+_1 \cup T^-_1$),
    \item the first (resp.\ second) counter is correctly decremented when using a
        transition from $T^-_1$ (resp.\ $T^-_2$) and the current value is not $0$.
\end{itemize}

For example we define $\A_5$ and $\B_5$ to check if the decrements of the first counter are correct using the gadget automaton $\Dd(x,y,z)$ in Figure~\ref{figure:undec-autA2}.
Let $x = \frac{1}{2}$, $y = 1$ and $z = \frac{1}{4}$. We define $\B_5$ as
$\Cc(x,x,x)$ and $\A_5$ as a weighted sum of $\Dd(x,y,z)$ and $\Dd(x,z,y)$ with
weights $(\frac{1}{2},\frac{1}{2})$. 

\begin{figure}[!htbp]
\begin{center}
    \begin{tikzpicture}[scale=0.5]
    \tikzset{
        initial text={\small{$x$}},
       	accepting/.style={accepting by arrow},
    }

\node[state, initial] (q_1) at (-10,0) {};
\node[state] (q_2) at (-3,0) {};
\node[state,accepting] (q_3) at (3,0) {};
\node[state, accepting] (q_4) at (10,0) {};
\node[state, initial] (q) at (-6.5,-3) {};

\path[trans]     (q_1) edge [loop above] node         {\small{$T:\frac{1}{2}$}} ();
\path[trans]     (q_1) edge [loop below] node         {\small{$a,b:1$}} ();
\path[trans]     (q_1) edge node [sloped, above] {\small{$T:\frac{1}{2}$}}(q);
\path[trans]     (q) edge  node [sloped, above] {\small{$a:y$}}(q_2);
\path[trans]     (q_2) edge [loop above] node         {\small{$a:y$}} ();
\path[trans]     (q_2) edge [loop below] node         {\small{$b:1$}} ();
\path[trans]     (q_2) edge  node [above] {\small{$T^-_1:z$}}(q_3);
\path[trans]     (q_3) edge [loop above] node         {\small{$a:z$}} ();
\path[trans]     (q_3) edge [bend left] node [above] {\small{$T,b:1$}}(q_4);
\path[trans]     (q_4) edge [loop above] node         {\small{$\alphabet:1$}} ();
    \end{tikzpicture}
\end{center}
\caption{\label{figure:undec-autA2} Gadget automaton $\Dd(x,y,z)$ used to check if
the first counter is decremented properly.}
\end{figure}

For all these automata $\Aone$, $\Bone$, \ldots $\Az$, $\Bz$, we have that for every proper word~$w$, $[\![\A_i]\!](w) = [\![\B_i]\!](w)$ if the corresponding increments or decrements are properly performed, and $[\![\A_i]\!](w) > [\![\B_i]\!](w)$ otherwise. We remark again that, in that case, for all $i$, we have $[\![\A_i]\!](w) \geq [\![\B_i]\!](w) + \left(\frac{1}{2}\right)^{2(|w|+1)}$. Note also that all the automata constructed above are linearly ambiguous (the only non-deterministic choices are in the first states when reading~$T$). 

Let us define $\A$ (resp.\ $\B$) as the weighted sum of the above automata
computing the function $\frac{7}{13}\funcAzero + \frac{1}{13}\funcAone + \cdots
+ \frac{1}{13}\funcAz$ (resp.\ $\frac{1}{13}\funcBone + \cdots +
\frac{1}{13}\funcBz$). The PA $\A$ and $\B$ are linearly ambiguous.

\begin{fact}
    For every word $w$ that does not represent the halting execution exactly one of
    the two cases below applies:
    \begin{itemize}
        \item either $w$ is not proper, $\funcAzero(w) = 1$, and $\funcA(w)
            \geq \frac{7}{13} > \frac{6}{13} \geq \funcB(w)$;
        \item or $w$ encodes an execution where the counters are not properly valued
            in at least one place, and in that case, there exists an $i$ such that $\A_i(w)
            > \B_i(w)$ (and $\geq$ for the other $i$). We thus obtain that $\funcA(w) >
            \funcB(w)$.
    \end{itemize}
\end{fact}

Remark that in the case when $\funcA(w) > \funcB(w)$ we have $\funcA(w) \geq \funcB(w) + \frac{1}{13}\left(\frac{1}{2}\right)^{2(|w|+1)}$.

\begin{fact}
    If the two-counter machine halts then for the word representing the halting
    execution we have that $\funcA(w) = \funcB(w)$.
\end{fact}

\subsubsection*{Step 2}

As noticed previously, by construction of $\A$ and $\B$, we have that:
\begin{itemize}
    \item there exists a word $w$ such that $\funcA(w) > \funcB(w)$ if and only
        if
    \item there exists a word $w$ such that $\funcA(w) \geq \funcB(w) +
        \frac{1}{13}\left(\frac{1}{2}\right)^{2(|w|+1)}$  if and only if
    \item there exists a word $w$ such that
    	\begin{equation}\label{eq:primed-autos}
        \frac{1}{2}\funcA(w) \quad \geq \quad 
        \frac{1}{2}\funcB(w) + \frac{1}{26}\left(\frac{1}{2}\right)^{2(|w|+1)}.
        \end{equation}
\end{itemize}
Both functions in the last inequality are computed by linearly ambiguous
PA. To obtain the weighted sum of the right hand side we need to construct an automaton that outputs, for every word $w$, the probability $ \frac{1}{13}\left(\frac{1}{2}\right)^{2(|w|+1)}$. This is easy to obtain by slightly modifying the example in Figure~\ref{fig:too-much-ambiguity}. Notice that this is an unambiguous automaton and its complement is linearly ambiguous.

For the rest of this section, we denote by $\Ap$ and $\Bp$ the PA for
the right and left sides of the inequality~\eqref{eq:primed-autos}, respectively. 

\subsubsection*{Step 3}

We will now argue that $1-\funcA$, $1-\funcB$, $1-\funcAp$, $1-\funcBp$ are also
computed by linearly ambiguous PA.
In Figure~\ref{figure:complement}, we show the complement automaton of
the gadget $\Cc(x,y,z)$ from Figure~\ref{figure:undec-autA}. We show that the so-obtained
automaton is linearly ambiguous. In the end we will conclude the argument using the property that all components $\A_i$ and $\B_i$ for $i > 0$ are weighted sums of such gadgets.

\begin{figure}[!htbp]
\begin{center}
    \begin{tikzpicture}[scale=0.5]
    \tikzset{
        initial text={\small{$x$}},
	    accepting/.style={accepting by arrow}
    }

    \node[state, initial, accepting] (q_1) at (-8,0) {$p$};
    \node[state, initial, accepting] (q_2) at (-3,0) {$q$};
    \node[state] (q_3) at (3,0) {};
    \node[state] (q_4) at (8,0) {};
    \node[state, accepting,initial,initial text={\small{$1-2x$}}] (sink)
        at (3,-6) {$\bot$};

    \path[trans]
    (q_1) edge [loop above] node         {\small{$T:\frac{1}{2}$}} ();
    \path[trans]
    (q_1) edge [loop below] node         {\small{$a,b:1$}} ();
    \path[trans]
    (q_1) edge [bend left] node [above] {\small{$T:\frac{1}{2}$}}(q_2);
    \path[trans]     (q_2) edge [loop above] node         {\small{$a:y$}} ();
    \path[trans]     (q_2) edge [loop below] node         {\small{$b:1$}} ();
    \path[trans]     (q_2) edge [bend left]  node [above] {\small{$T_1^{+}:y$}}(q_3);
    \path[trans]     (q_3) edge [loop above] node         {\small{$a:z$}} ();
    \path[trans]     (q_3) edge [bend left] node [above] {\small{$T,b:1$}}(q_4);

    \path[trans] (q_2) edge
    node [align=center,sloped, above]
    {\small{$a, T_1^{+}:1-y$}\\\small{$T\setminus T_1^{+}:1$}} (sink);
    \path[trans]		(sink) edge[loop above] node {\small{$\Sigma:1$}} ();
    \end{tikzpicture}
\end{center}
\caption{\label{figure:complement} Complement automaton of $\mathcal{C}(x,y,z)$
after trimming.}
\end{figure}
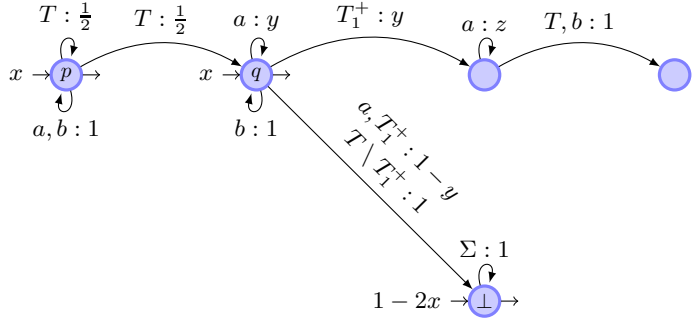

First, notice that we can trim the automaton to have only three states $p,q$ and $\bot$ because the remaining states are not accepting and it is not possible to reach an accepting state from them.
Recall that the ambiguity of an automaton relies only on its underlying structure (Section ~\ref{sec:ambiguity}).
Hence we can focus on the Boolean automaton (i.e.\ without probabilities) from Figure~\ref{figure:complement2} and analyze its ambiguity.
In the trimmed automaton there are only two places with non-deterministic choices: when reading a letter in $T$ from $p$ the automaton can either remain in $p$ or move to $q$; when reading $a$ from $q$ the automaton can either remain in $q$ or move to $\bot$.

Notice that all three states are both initial and accepting.
Let us decompose the set of accepting runs of the automaton on a word $w$ depending on where the run starts and where the run ends, which is 9 cases in total.
We focus only on the case for run starting in $p$ and ending in $\bot$; the remaining cases obviously provide at most a linear number of runs.
The automaton can move from $p$ to $q$ only 
when reading an element of $T$. Fix a word $w$ and consider positions $i$ and $i'$
such that $t_i \in T$ and $i'$ is maximal such that positions between $i$ and $i'$
have labels from $\Sigma \setminus T$ (i.e. position $i'$ is labelled with a letter
from $T$ or is equal to $|w|+1$). We show that there are at most
$f(i) \defequals i'-i + 1$ accepting runs starting from $p$ and ending in $q$ after
reading $t_i$.
This is because in state $q$, when reading an element from $T$, the automaton
has to move to $\bot$. Hence the number of all accepting runs from $p$ to $\bot$
is bounded by the sum of all $f(i)$ through all positions $1 \leq i \leq |w|$ such that $t_i \in T$.
We conclude that the automaton is linearly ambiguous.

\begin{figure}[!htbp]
\begin{center}
\begin{tikzpicture}[scale=0.5]
\tikzset{
	initial text={},
	accepting/.style={accepting by arrow},
}

\node[state, initial, accepting] (q_1) at (-8,0) {$p$};
\node[state, initial, accepting] (q_2) at (-3,0) {$q$};
\node[state, accepting,initial] (sink) at (2,0) {$\bot$};

\path[trans]	(q_1) edge [loop above] node {\small{$\Sigma$}} ();
\path[trans]	(q_2) edge [loop above] node {\small{$a,b$}} ();
\path[trans]	(sink) edge [loop above] node {\small{$\Sigma$}} ();

\path[trans]	(q_1) edge [bend left] node[above] {\small{$T$}} (q_2);
\path[trans]	(q_2) edge [bend left] node[above] {\small{$\Sigma \setminus \{b\}$}} (sink);
\end{tikzpicture}
\end{center}
\caption{\label{figure:complement2} Trimmed version of the automaton from Figure~\ref{figure:complement}; the transition probabilities have been omitted.}
\end{figure}
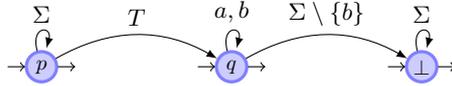

To conclude our argument about the ambiguity of all constructed automata, recall that
automata $\A$, $\B$, $\Ap$, and $\Bp$ are constructed as weighted sums of
PA obtained from $\mathcal{A}_i$ and $\mathcal{B}_j$, which are weighted sums of gadgets like $\Cc(x,y,z)$. 
For example the function $1-\funcA$ is
\begin{align*}
    1-\funcA \quad & = \quad 1- \left(\frac{7}{13}\funcAzero + \frac{1}{13}\funcAone + \cdots +
    \frac{1}{13}\funcAz\right) \\
    & = \quad \frac{7}{13}(1 - \funcAzero) + \frac{1}{13}(1 - \funcAone) + \cdots +
    \frac{1}{13}(1 - \funcAz) \\
\end{align*}
It is thus sufficient to complement each member of the sum. Showing that all these
complements are still linearly ambiguous follows the ideas given above for the
automaton in Figure~\ref{figure:complement}.

\section{Conclusion}
\label{sec:conclusions}
In this work we have shown that the containment problem for PA
is decidable if one of the automata is finitely ambiguous and the other
one is unambiguous. Interestingly, for one of the two cases, our proposed
algorithm uses a satisfiability oracle for a theory whose decidability is
equivalent to a weak form of Schanuel's conjecture. We have complemented our
decidability results with a proof of undecidability for the case when the
given automata are linearly ambiguous.

Decidability of the containment problem when both automata are allowed to be
finitely ambiguous remains open. One way to tackle it is to study
generalizations of the IP+EXP problem introduced in
Section~\ref{sec:decidability1vsk}. This problem asks whether there exists $\xx
\in \mathbb{N}^n$ such that $f(\xx) < 1$ and $M\xx  < \cc$ for a given function
$f$ defined using exponentiations, a given matrix $M$, and vector $\cc$.  A
natural way to extend the latter would be to ask that $f(\xx) < g(\xx)$, where
$g$ is obtained in a similar way as~$f$.  The main obstacle, when trying to
generalize our decidability proof for that problem, is that we lack a
replacement for the cone $C$  needed  in order to obtain a result similar to
Lemma~\ref{lemma:orthog} using the Minkowski-Weyl decomposition.

\section*{Acknowledgements}
This work was supported by the EPSRC grant EP/P020992/1 and the EPSRC fellowship
EP/N008197/1.  R. Lazi\'{c} was also supported by a Leverhulme Trust Research
Fellowship RF-2017-579; F. Mazowiecki, by the French National Research Agency
(ANR) in the frame of the ``Investments for the future'' Programme IdEx Bordeaux
(ANR-10-IDEX-03-02); G. A. P\'erez, by an F.R.S.-FNRS Aspirant fellowship and an
FWA postdoc fellowship.

We thank Shaull Almagor and Isma??l Jecker for some helpful remarks.

\bibliographystyle{elsarticle-num}
\bibliography{refs}

\begin{thebibliography}{10}
\expandafter\ifx\csname url\endcsname\relax
  \def\url#1{\texttt{#1}}\fi
\expandafter\ifx\csname urlprefix\endcsname\relax\def\urlprefix{URL }\fi
\expandafter\ifx\csname href\endcsname\relax
  \def\href#1#2{#2} \def\path#1{#1}\fi

\bibitem{rabin63}
M.~O. Rabin, Probabilistic automata, Information and Control 6~(3) (1963)
  230--245.
\newblock \href {https://doi.org/10.1016/S0019-9958(63)90290-0}
  {\path{doi:10.1016/S0019-9958(63)90290-0}}.

\bibitem{puterman05}
M.~L. Puterman, {{Markov} Decision Processes}, Wiley-Interscience, 2005.

\bibitem{baum1966}
L.~E. Baum, T.~Petrie, Statistical inference for probabilistic functions of
  finite state markov chains, Ann. Math. Statist. 37~(6) (1966) 1554--1563.
\newblock \href {https://doi.org/10.1214/aoms/1177699147}
  {\path{doi:10.1214/aoms/1177699147}}.

\bibitem{rn10}
S.~J. Russell, P.~Norvig, Artificial Intelligence - {A} Modern Approach (3.
  internat. ed.), Pearson Education, 2010.

\bibitem{learning-survey}
L.~P. Kaelbling, M.~L. Littman, A.~W. Moore, Reinforcement learning: {A}
  survey, Journal of Artificial Intelligence Research 4 (1996) 237--285.
\newblock \href {https://doi.org/10.1613/jair.301}
  {\path{doi:10.1613/jair.301}}.

\bibitem{vardi85}
M.~Y. Vardi, Automatic verification of probabilistic concurrent finite-state
  programs, in: 26th Annual Symposium on Foundations of Computer Science,
  Portland, Oregon, USA, 21-23 October 1985, {IEEE} Computer Society, 1985, pp.
  327--338.
\newblock \href {https://doi.org/10.1109/SFCS.1985.12}
  {\path{doi:10.1109/SFCS.1985.12}}.

\bibitem{knpq10}
M.~Z. Kwiatkowska, G.~Norman, D.~Parker, H.~Qu, Assume-guarantee verification
  for probabilistic systems, in: J.~Esparza, R.~Majumdar (Eds.), Tools and
  Algorithms for the Construction and Analysis of Systems, 16th International
  Conference, {TACAS} 2010, Held as Part of the Joint European Conferences on
  Theory and Practice of Software, {ETAPS} 2010, Paphos, Cyprus, March 20-28,
  2010. Proceedings, Vol. 6015 of Lecture Notes in Computer Science, Springer,
  2010, pp. 23--37.
\newblock \href {https://doi.org/10.1007/978-3-642-12002-2_3}
  {\path{doi:10.1007/978-3-642-12002-2_3}}.

\bibitem{FengHKP11}
L.~Feng, T.~Han, M.~Z. Kwiatkowska, D.~Parker, Learning-based compositional
  verification for synchronous probabilistic systems, in: T.~Bultan, P.~Hsiung
  (Eds.), Automated Technology for Verification and Analysis, 9th International
  Symposium, {ATVA} 2011, Taipei, Taiwan, October 11-14, 2011, Vol. 6996 of
  Lecture Notes in Computer Science, Springer, 2011, pp. 511--521.
\newblock \href {https://doi.org/10.1007/978-3-642-24372-1_40}
  {\path{doi:10.1007/978-3-642-24372-1_40}}.

\bibitem{PalemA13}
K.~V. Palem, L.~Avinash, Ten years of building broken chips: The physics and
  engineering of inexact computing, {ACM} Transactions on Embedded Computing
  Systems 12~(2s) (2013) 87:1--87:23.
\newblock \href {https://doi.org/10.1145/2465787.2465789}
  {\path{doi:10.1145/2465787.2465789}}.

\bibitem{YakaryilmazS11}
A.~Yakaryilmaz, A.~C.~C. Say, Unbounded-error quantum computation with small
  space bounds, Information and Computation 209~(6) (2011) 873--892.
\newblock \href {https://doi.org/10.1016/j.ic.2011.01.008}
  {\path{doi:10.1016/j.ic.2011.01.008}}.

\bibitem{GieseBPRIWC11}
H.~Giese, N.~Bencomo, L.~Pasquale, A.~J. Ramirez, P.~Inverardi,
  S.~W{\"{a}}tzoldt, S.~Clarke, Living with uncertainty in the age of runtime
  models, in: N.~Bencomo, R.~B. France, B.~H.~C. Cheng, U.~Assmann (Eds.),
  Models@run.time - Foundations, Applications, and Roadmaps [Dagstuhl Seminar
  11481, November 27 - December 2, 2011], Vol. 8378 of Lecture Notes in
  Computer Science, Springer, 2014, pp. 47--100.
\newblock \href {https://doi.org/10.1007/978-3-319-08915-7_3}
  {\path{doi:10.1007/978-3-319-08915-7_3}}.

\bibitem{mpr02}
M.~Mohri, F.~Pereira, M.~Riley, Weighted finite-state transducers in speech
  recognition, Computer Speech {\&} Language 16~(1) (2002) 69--88.
\newblock \href {https://doi.org/10.1006/csla.2001.0184}
  {\path{doi:10.1006/csla.2001.0184}}.

\bibitem{Schutzenberger61b}
M.~P. Sch{\"{u}}tzenberger, On the definition of a family of automata,
  Information and Control 4~(2-3) (1961) 245--270.

\bibitem{Tzeng92}
W.~Tzeng, A polynomial-time algorithm for the equivalence of probabilistic
  automata, {SIAM} J. Comput. 21~(2) (1992) 216--227.

\bibitem{fgko15}
N.~Fijalkow, H.~Gimbert, E.~Kelmendi, Y.~Oualhadj, Deciding the value 1 problem
  for probabilistic leaktight automata, Logical Methods in Computer Science
  11~(2) (2015).
\newblock \href {https://doi.org/10.2168/LMCS-11(2:12)2015}
  {\path{doi:10.2168/LMCS-11(2:12)2015}}.

\bibitem{csvb15}
R.~Chadha, A.~P. Sistla, M.~Viswanathan, Y.~Ben, Decidable and expressive
  classes of probabilistic automata, in: A.~M. Pitts (Ed.), Foundations of
  Software Science and Computation Structures - 18th International Conference,
  FoSSaCS 2015, Held as Part of the European Joint Conferences on Theory and
  Practice of Software, {ETAPS} 2015, London, UK, April 11-18, 2015.
  Proceedings, Vol. 9034 of Lecture Notes in Computer Science, Springer, 2015,
  pp. 200--214.
\newblock \href {https://doi.org/10.1007/978-3-662-46678-0_13}
  {\path{doi:10.1007/978-3-662-46678-0_13}}.

\bibitem{frw17}
N.~Fijalkow, C.~Riveros, J.~Worrell, Probabilistic automata of bounded
  ambiguity, in: R.~Meyer, U.~Nestmann (Eds.), 28th International Conference on
  Concurrency Theory, {CONCUR} 2017, September 5-8, 2017, Berlin, Germany,
  Vol.~85 of LIPIcs, Schloss Dagstuhl - Leibniz-Zentrum fuer Informatik, 2017,
  pp. 19:1--19:14.
\newblock \href {https://doi.org/10.4230/LIPIcs.CONCUR.2017.19}
  {\path{doi:10.4230/LIPIcs.CONCUR.2017.19}}.

\bibitem{fijalkow17}
N.~Fijalkow, Undecidability results for probabilistic automata, {SIGLOG} News
  4~(4) (2017) 10--17.
\newblock \href {https://doi.org/10.1145/3157831.3157833}
  {\path{doi:10.1145/3157831.3157833}}.

\bibitem{paz14}
A.~Paz, Introduction to probabilistic automata, Academic Press, 1971.

\bibitem{BlondelC03}
V.~D. Blondel, V.~Canterini, Undecidable problems for probabilistic automata of
  fixed dimension, Theory Comput. Syst. 36~(3) (2003) 231--245.

\bibitem{AkshayAOW15}
S.~Akshay, T.~Antonopoulos, J.~Ouaknine, J.~Worrell, Reachability problems for
  markov chains, Inf. Process. Lett. 115~(2) (2015) 155--158.

\bibitem{cl89}
A.~Condon, R.~J. Lipton, On the complexity of space bounded interactive proofs
  (extended abstract), in: 30th Annual Symposium on Foundations of Computer
  Science, Research Triangle Park, North Carolina, USA, 30 October - 1 November
  1989, {IEEE} Computer Society, 1989, pp. 462--467.
\newblock \href {https://doi.org/10.1109/SFCS.1989.63519}
  {\path{doi:10.1109/SFCS.1989.63519}}.

\bibitem{Goldreich05}
O.~Goldreich,
  \href{http://eccc.hpi-web.de/eccc-reports/2005/TR05-018/index.html}{On
  promise problems (a survey in memory of shimon even {[1935-2004])}},
  Electronic Colloquium on Computational Complexity {(ECCC)}~(018) (2005).
\newline\urlprefix\url{http://eccc.hpi-web.de/eccc-reports/2005/TR05-018/index.html}

\bibitem{DerksenJK05}
H.~Derksen, E.~Jeandel, P.~Koiran, Quantum automata and algebraic groups, J.
  Symb. Comput. 39~(3-4) (2005) 357--371.

\bibitem{WeberS91}
A.~Weber, H.~Seidl, On the degree of ambiguity of finite automata, Theor.
  Comput. Sci. 88~(2) (1991) 325--349.

\bibitem{mw96}
A.~Macintyre, A.~J. Wilkie, On the decidability of the real exponential field,
  in: P.~Odifreddi (Ed.), Kreiseliana. About and Around Georg Kreisel, AK
  Peters, 1996, pp. 441--467.

\bibitem{ws91}
A.~Weber, H.~Seidl, On the degree of ambiguity of finite automata, Theoretical
  Computer Science 88~(2) (1991) 325--349.
\newblock \href {https://doi.org/10.1016/0304-3975(91)90381-B}
  {\path{doi:10.1016/0304-3975(91)90381-B}}.

\bibitem{Rackoff78}
C.~Rackoff, The covering and boundedness problems for vector addition systems,
  Theoretical Compututer Science 6 (1978) 223--231.
\newblock \href {https://doi.org/10.1016/0304-3975(78)90036-1}
  {\path{doi:10.1016/0304-3975(78)90036-1}}.

\bibitem{cohen93}
H.~Cohen, A course in computational algebraic number theory, Vol. 138 of
  Graduate texts in mathematics, Springer, 1993.

\bibitem{kp97}
L.~Khachiyan, L.~Porkolab, Computing integral points in convex semi-algebraic
  sets, in: 38th Annual Symposium on Foundations of Computer Science, {FOCS}
  '97, Miami Beach, Florida, USA, October 19-22, 1997, {IEEE} Computer Society,
  1997, pp. 162--171.
\newblock \href {https://doi.org/10.1109/SFCS.1997.646105}
  {\path{doi:10.1109/SFCS.1997.646105}}.

\bibitem{Colcombet}
T.~Colcombet, On distance automata and regular cost function, presented at the
  Dagstuhl seminar “Advances and Applications of Automata on Words and
  Trees” (2010).

\bibitem{abk11}
S.~Almagor, U.~Boker, O.~Kupferman, What's decidable about weighted automata?,
  in: T.~Bultan, P.~Hsiung (Eds.), Automated Technology for Verification and
  Analysis, 9th International Symposium, {ATVA} 2011, Taipei, Taiwan, October
  11-14, 2011. Proceedings, Vol. 6996 of Lecture Notes in Computer Science,
  Springer, 2011, pp. 482--491.
\newblock \href {https://doi.org/10.1007/978-3-642-24372-1_37}
  {\path{doi:10.1007/978-3-642-24372-1_37}}.

\bibitem{minsky67}
M.~L. Minsky, Computation: Finite and Infinite Machines, Prentice-Hall, 1967.

\end{thebibliography}

\end{document}